\documentclass[11pt,leqno]{article}
\usepackage{cmu-titlepage2}
\usepackage{fullpage}

\usepackage{amsthm}\newtheorem{theorem}{Theorem}\newtheorem{lemma}{Lemma}\newtheorem{proposition}{Proposition}
\usepackage{amsmath,amsfonts,amssymb,url,verbatim,stmaryrd,ellipsis}
\usepackage{epsfig,color}
\usepackage{amsmath,amsfonts,amssymb,url,ellipsis}
\usepackage{hyperref}

\let\Paragraph=\paragraph \renewcommand{\paragraph}[1]{\Paragraph{{#1}.}}

\newcommand*{\cross}{\times}
\newcommand*{\N}{\mathbb{N}}
\newcommand*{\R}{\mathbb{R}}
\newcommand*{\E}{\mathbb{E}}
\newcommand*{\opt}{\mathsf{opt}}

\newcommand*{\nothing}{\mathsf{N}}

\newcommand*{\subseq}{\sqsubset}
\newcommand*{\set}[2]{\{\,{#1}~|~{#2}\,\}}
\newcommand*{\cons}{{:}}

\newcommand{\setless}{-}

\newcommand*{\mathsc}[1]{\textsc{{#1}}}

\newcommand*{\beforenothing}{\mathsf{active}}
\newcommand*{\substrg}{\prec}
\newcommand*{\behv}{\mathsf{behv}}

\newcommand*{\fix}{\mathsf{fix}}
\newcommand*{\strg}{\mathsf{strg}}

\newcommand*{\e}{e}
\newcommand*{\be}{b}

\newcommand*{\email}[1]{{\url{#1}}}

\DeclareMathOperator*{\preargmax}{\mathsf{argmax}}
\newcommand*{\argmax}{\preargmax\limits}

\title{On the Semantics of Purpose Requirements\\ in Privacy Policies}
\date{February 18, 2011}
\author{Michael Carl Tschantz \and Anupam Datta \and  Jeannette M. Wing}
\trnumber{CMU-CS-11-102}
\abstract{%
Privacy policies often place requirements on the purposes for which a
governed entity may use personal information.  For example,
regulations, such as HIPAA, require that hospital employees use
medical information for only certain purposes, such as treatment.
Thus, using formal or automated methods for enforcing privacy policies
requires a semantics of \emph{purpose requirements} to determine
whether an action is \emph{for} a purpose or not.
We provide such a semantics using a formalism based on
\emph{planning}.  We model planning using a modified version of Markov
Decision Processes, which exclude redundant actions for a formal definition
of \emph{redundant}.  We use the model to formalize when a sequence of
actions is \emph{only for} or \emph{not for} a purpose.  This
semantics enables us to provide an algorithm for automating auditing,
and to describe formally and compare rigorously previous enforcement
methods.}

\support{This research was supported by the US Army Research Office under grant
numbers W911NF0910273 and DAAD-190210389. The views and conclusions
contained in this document are those of the authors and should not be
interpreted as representing the official policies, either expressed or
implied, of any sponsoring institution, the U.S.\ government or any other
entity.}

\history{\noindent This manuscript was submitted to the 24th IEEE Computer Security Foundations Symposium.}

\begin{document}

\maketitle

\section{Introduction}
\label{sec:intro}


\emph{Purpose} is a key concept for privacy policies.
For example, the European Union requires that~\cite{eu95directive}:
\begin{quote}
Member States shall provide that personal data must be [\ldots] collected for specified, explicit and legitimate purposes and not further processed in a way incompatible with those purposes. 
\end{quote}
The United States also has laws placing purpose requirements on information in some domains such as HIPAA~\cite{hipaa} for medical information and the Gramm-Leach-Bliley Act~\cite{glb} for financial records.
These laws and best practices motivate organizations to discuss in their privacy policies the purposes for which they will use information.

Some privacy policies warn users that the policy provider may use certain information for certain purposes.  For example,
the privacy policy of a medical provider states, ``We may disclose your [protected health information] for public health activities and purposes [\ldots]''~\cite{washington}.
Such warnings do not constrain the behavior of the policy provider.

Other policies that prohibit using certain information for a purpose do constrain the behavior of the policy provider.  Examples include the privacy policy of Yahoo!\ Email, which states that ``Yahoo!'s practice is \emph{not} to use the content of messages stored in your Yahoo!\ Mail account \emph{for} marketing purposes''~\cite[emphasis added]{yahoo-email-policy}.

Some policies even limit the use of certain information to an explicit list of purposes.  The privacy policy of The Bank of America states, ``Employees are authorized to access Customer Information \emph{for} business purposes \emph{only}.''~\cite[emphasis added]{bankam}.  The HIPAA Privacy Rule~\cite{hipaa} requires that covered entities (e.g., health care providers and business partners) only use or disclose protected health information about a patient with that patient's written authorization or:
\begin{quote}
[\ldots] for the following purposes or situations: (1) To the Individual [\ldots]; (2) Treatment, Payment, and Health Care Operations; (3) Opportunity to Agree or Object; (4) Incident to an otherwise permitted use and disclosure; (5) Public Interest and Benefit Activities; and (6) Limited Data Set for the purposes of research, public health or health care operations.
\end{quote}

These examples show that verifying that an organization obeys a privacy policy requires a semantics of \emph{purpose requirements}.  In particular, enforcement requires the ability to determine that the organization under scrutiny obeys at least two classes of purpose requirements.  As shown in the example rule from Yahoo!, the first requirement is that the organization does \emph{not} use certain sensitive information \emph{for} a given purpose.  The second, as the example rule from HIPAA shows, is that the organization uses certain sensitive information \emph{only for} a given list of purposes.  We call the first class of requirements \emph{prohibitive} (not-for) and the second class \emph{restrictive} (only-for). 
Each class requires determining whether the organization's behavior is \emph{for} a purpose or not, but they differ in whether this indicates a violation or compliance, respectively.

For example, consider a physician accessing a medical record.  Under the HIPAA Privacy Rule, the physician may access the record only for certain purposes such as treatment, research, and billing.  Thus, for an auditor (either internal or external) to determine whether the physician has obeyed the Privacy Rule requires the auditor to determine the purposes for which the physician accessed the record.  The auditor's ability to determine the purposes behind actions is limited since the auditor can only observe the behavior of the physician.  As a physician may perform the exact same actions for different purposes, the auditor can never be sure of the purposes behind an action.  However, if the auditor determines that the record access could not have possibly been for any of the purposes allowed under the Privacy Rule, then the auditor knows that the physician violated the policy.

Manual enforcement of these privacy policies is labor intensive and
error prone.  Thus, to reduce costs and make their operations more
trustworthy, organizations would like to automate the enforcement of
the privacy policies governing their operations; tool support for this 
activity is beginning to emerge in the market. For example, Fair
Warning offers automated services for the detection of privacy
breaches in a hospital setting~\cite{fairwarning}.  
Meanwhile, previous research has purposed formal methods to enforce purpose requirements~\cite{aksx02hippocratic,bbl05purpose,ha05language,a07beyond,bl08purpose,pgy08dynamic,jss09enforcing,nblbkkt10privacy,kwb11conditional}.  

However, each of these endeavors start by assuming that actions or sequences of actions are labeled with the purposes they are \emph{for}. They avoid analyzing the meaning of \emph{purpose} and provide no method of performing this labeling other than through intuition alone.  The absence of a formal semantics to guide this determination has hampered the development of methods for ensuring policy compliance.  Such a definition would provide insights into how to develop tools that identify suspicious accesses in need of detailed auditing and algorithms for determining which purposes an action could possibly be for.  Such a definition would also show which enforcement approaches are most accurate.  More fundamentally, such a definition could frame the scientific basis of a societal and legal understanding of purpose and of privacy policies that use the notion of purpose.  Such a foundation can, for example, guide implementers as they codify in software an organization's interpretation of internal and government-imposed privacy policies.

\subsection{Solution Approach}
The goal of this work is to study the meaning of \emph{purpose} in the context of enforcing privacy policies and propose formal definitions suitable for automating the enforcement of purpose requirements.  
Since post-hoc auditing provides the perspective often required to determine the purpose of an action, we focus on automated auditing. 
However, we believe our semantics is applicable to other formal methods and may also clarify informal reasoning.

We find that \emph{planning} is central to the meaning of purpose.  We see the role of planning in the definition of the sense of the word ``purpose'' most relevant to our work~\cite{oed}:
\begin{quote}
The object for which anything is done or made, or for which it exists; the result or effect intended or sought; end, aim.
\end{quote}
Similarly, work on cognitive psychology calls purpose ``the central determinant of behavior''~\cite[p19]{dkp96cognitive}.
If our auditors are concerned with rational auditees (the person or organization being audited), then we may assume the auditee uses a plan to determine what actions it will perform in its attempt to achieve its purposes.   We (as have philosophers~\cite{t66action}) conclude that if an auditee selects to perform an action $a$ while planning to achieve the purpose $p$, then the auditee's action $a$ is \emph{for the purpose} $p$.  
In this paper, we make these notions formal.

\subsection{Overview of Contributions}

We first present an example that illustrates key factors in determining whether an action is for a purpose or not.
We find that the auditor should model the auditee as an agent that interacts with an
\emph{environment model}.  The environment model shows how the actions
the auditee can perform affect the state of the environment.  It also
models how well each state satisfies each purpose that the modeled
auditee might possibly find motivating.  Limiting consideration to
one purpose, the environment model becomes a Markov
Decision Process (MDP) where the degree of satisfaction of that purpose is the reward function of the MDP.  If the auditee is
motivated to act by only that purpose, then the auditee's actions
must correspond to an optimal \emph{plan} for this MDP and these actions are \emph{for} that purpose.  Additionally, we use a stricter definition of optimal than standard MDPs to reject redundant actions that neither decrease nor increase the total reward.  We formalize this model in Section~\ref{sec:formalism}.

For example, consider a physician ordering a medical test and an auditor attempting to determine whether the physician could have ordered this test for the purpose of treatment
(and is therefore in compliance with the  HIPAA Privacy Rule).   The auditor would examine an MDP modeling the physician's environment with the quality of treatment as the reward function to be optimized.  If no optimal plans for this MDP involve ordering the test, then the auditor can conclude definitively that the physician did not order the test for treatment.

We make this auditing process formal in Section~\ref{sec:audit} where we discuss the ramifications of the auditor only observing the behaviors of the auditee and not the underlying planning process of the auditee that resulted in these behaviors.  
We show that in some circumstances, the auditor can still acquire enough information to determine that the auditee violated the privacy policy.
To do so, the auditor must first use our MDP model to construct all the possible behaviors that the privacy policy allows and then compare it with all the behaviors of the auditee that could have resulted in the observed auditing log.  
Section~\ref{sec:algo} presents an algorithm for auditing based on our formal definitions, illustrating the relevance of our work.

The semantics discussed thus far is sufficient to put the previous work on enforcing privacy policies on firm semantic ground.  In Section~\ref{sec:compare}, we do so and discuss the strengths and weaknesses of each such approach.  In particular, we find that each approach may be viewed as a method of enforcing the policy given the set of all possible allowed behaviors, an intermediate result of our analysis.
  We compare the previous auditing approaches, which differ in their trade-offs between auditing complexity and accuracy of representing this set of behaviors.

Most auditees are actually interested in multiple purposes and select plans that simultaneously satisfy as many of the desired purposes as possible.  Handling the interactions between purposes complicates our semantics.  In particular, actions selected by a single plan may be for different purposes.
In Section~\ref{sec:multi-purposes}, we present examples showing when our semantics can extend to handle multiple purposes and when difficulties arise in determining which purposes an action is for when an auditee is attempting to satisfy various purposes at once. 
Currently, the state-of-the-art in the understanding of human planning limits our abilities to improve upon our semantics.  
However, as this understanding improves, one may replace our MDP-like formalism with more detailed ones while retaining our 
general framework of defining purpose in terms of planning.

We end by discussing other related work, future work, and conclusions.
Our contributions include:
\begin{itemize}
\item The first semantic formalism of when a sequence of actions is for a purpose,
\item An auditing algorithm for this formalism,
\item The resituating of previous policy enforcement methods in our formalism and a comparative study of their expressiveness, and 
\item The first attempt to formally consider the effects on auditing caused by interactions among multiple purposes.
\end{itemize}

Although motivated by our goal to formalize the notions of \emph{use}
and \emph{purpose} prevalently found in privacy policies, our
work is more generally applicable to a broad range of
policies, such as fiscal policies governing travel reimbursement.

\section{Motivation of Our Approach}
\label{sec:motivation}

We start with an informal example that suggests that \emph{an action is for a purpose if the action is part of a plan for achieving that purpose}.
Consider a physician
working at a hospital who, as a specialist, also owns a private practice that tests for bone damage using a novel technique for extracting information from X-ray images.  After seeing a patient and taking an X-ray, the physician forwards the patient's
medical record including the X-ray to his private practice to apply
this new technology.  As this action entails the transmission of
protected health information, the physician will have violated HIPAA
if this transmission is not for one of the purposes HIPAA allows.  The
physician would also run afoul of the hospital's own policies
governing when outside consultations are permissible unless this
action was for a legitimate purpose.  
Finally, the patient's insurance will only reimburse the costs associated with this consultation if a medical reason (purpose) exists for them.  
The physician claims that this consultation was for reaching a diagnosis.  As such, it is for the purpose of treatment and, therefore, allowed under each of these policies.  The hospital auditor, however, has selected this
action for investigation since the physician's
 making a referral to his own private practice makes the alternate motivation of profit possible.

Whether or not the physician violated these policies depends
upon details not presented in the above description.  For example, we
would expect the auditor to ask questions such as:
(1) Was the test relevant to the patient's condition? 
(2) Did the patient benefit medically from having the test?
(3) Was this the best option for the patient?
We will introduce these details as we introduce each of the factors relevant to the purposes behind the physician's actions.

\paragraph{States and Actions}
Sometimes the purposes for which an action is taken depend upon the previous actions and the state of the system.  
In the above example, whether or not the test is relevant depends upon the condition of the patient, that is, the state that the patient is in.

While an auditor could model the act of transmitting the record as two (or more) different actions based upon the state of the patient, 
modeling two concepts with one formalism could introduce errors.
  A better approach is to model the state of the system.  The state captures the context in which the physician takes an action and allows for the purposes of an action to depend upon the actions that precede it.

The physician's own actions also affect the state of the system and, thus, the purposes for which his actions are.  For example, had the physician transmitted the patient's medical record before taking the X-ray, then the transmission could not have been for treatment since the physician's private practice only operates on X-rays and would have no use for the record without the X-ray.

The above example illustrates that when an action is for a purpose, the action is part of a sequence of actions that can lead to a state in which some goal associated with the purpose is achieved.  In the example, the goal is reaching a diagnosis.  Only when the X-ray is first added to the record is this goal reached.

\paragraph{Non-redundancy}
Some actions, however, may be part of such a sequence without actually
being for the purpose.  For example, suppose that the patient's X-ray
clearly shows the patient's problem.  Then, the physician can reach a
diagnosis without sending the record to the private practice.  Thus,
while both taking the X-ray and sending the medical record might be part of
a sequence of actions that leads to achieving a diagnosis, the
transmission does not actually contribute to achieving the diagnosis:
the physician could omit it and the diagnosis could still be reached.

From this example, it may be tempting to conclude that an action is
\emph{for} a purpose only if that action is \emph{necessary} to
achieve that purpose.  However, consider a physician who has a choice
between two specialists to whom to send the medical record and must do so to reach a diagnosis.  
In this scenario, the physician's sending the  record to the first specialist is not necessary since he could send it to the second.  Likewise, sending the record to the second specialist is not necessary.  Yet, the physician must send the record to one or the other specialist and that transmission will be for the purpose of diagnosis.  Thus, an action may be for a purpose without being necessary for achieving the purpose.

Rather than \emph{necessity}, we use the weaker notion of \emph{non-redundancy} found in work on the semantics of \emph{causation} (e.g.,~\cite{m74cement}).  Given a sequence of actions that achieves a goal, an action in it is \emph{redundant} if that sequence with that action removed (and otherwise unchanged) also achieves the goal.
An action is \emph{non-redundant} if removing that action from the sequence would result in the goal no longer being achieved.
  Thus, non-redundancy may be viewed as necessity under an otherwise fixed sequence of actions.

For example, suppose the physician decides to send the medical record to the first specialist.  Then, the sequence of actions modified by removing this action would not lead to a state in which a diagnosis is reached.  Thus, the transmission of the medical record to the first specialist is non-redundant.  However, had the X-ray revealed to the physician the diagnosis without needing to send it to a specialist, the sequence of actions that results from removing the transmission from the original sequence would still result in a diagnosis.  Thus, the transmission would be redundant.

\paragraph{Quantitative Purposes}
Above we implicitly presumed that the diagnosis from each specialist had equal quality.
This need not be the case.  Indeed, many purposes are actually fulfilled to varying degrees.  For example, the purpose of marketing is never completely achieved since there is always more marketing to do.  
Thus, we model a purpose by assigning to each state-action pair a number that describes how well that action fulfills that purpose when performed in that state.  We require that the physician selects the test that maximizes the quality of the diagnosis as determined by total purpose score accumulated over all his actions.

\paragraph{Probabilistic Systems}
The success of many medical tests and procedures is probabilistic.  For example, with some probability the physician's test may fail to reach a diagnosis.  The physician would still have transmitted the medical record for the purpose of diagnosis even if the test failed to reach one.  This possibility affects our semantics of purpose: now an action may be for a purpose even if that purpose is never achieved.  

To account for such probabilistic events, we model the environment in which the physician operates as probabilistic.  For an action to be for a purpose, we require that there be a non-zero probability of the purpose being achieved and that the physician attempts to maximize the expected reward.  In essence, we require that the physician attempts to achieve a diagnosis.  Thus, the auditee's \emph{plan} determines the purposes behind his actions rather than just the actions themselves.

\section{Planning for a Purpose}
\label{sec:formalism}

In this section, we present a formalism for planning that accounts for quantitative purposes, probabilistic systems and non-redundancy.
We start by modeling the environment in which the auditee operates as a Markov Decision Process (MDP)---a natural model for probabilistic systems. 
The reward function of the MDP quantifies the degree of satisfaction of a purpose upon taking an action from a state. 
If the auditee is motivated to action by only that purpose, then the auditee's actions
must correspond to an optimal \emph{plan} for this MDP and these actions are \emph{for} that purpose.  
We develop a stricter definition of optimal than standard MDPs, which we call NMDPs for \emph{Non-redundant MDP}, 
to reject redundant actions that neither decrease nor increase the total reward.
We end with an example illustrating the use of an NMDP to model an audited environment.

\subsection{Markov Decision Processes}
\label{sec:mdps}

An MDP may be thought of as a probabilistic automaton where transitions are labeled with a reward in addition to an action.  Rather than having accepting or goal states, the ``goal'' of a MDP is maximizing the total reward over time.

An MDP is a tuple $m = \langle \mathcal{Q}, \mathcal{A}, t, r, \gamma\rangle$ where $\mathcal{Q}$ is a set of states, $\mathcal{A}$ is a set of actions, $t : \mathcal{Q} \cross \mathcal{A} \to \mathcal{D}(\mathcal{Q})$ a transition function from a state and an action to a distribution over states (represented as $\mathcal{D}(\mathcal{Q})$), $r : \mathcal{Q} \cross \mathcal{A} \to \R$ a reward function, and $\gamma$ a discount factor such that $0 < \gamma < 1$.  For each state $q$ in $\mathcal{Q}$, the agent using the MDP to plan selects an action $a$ from $\mathcal{A}$ to perform.  Upon performing the action $a$ in the state $q$, the agent receives the reward $r(q,a)$.  The environment then transitions to a new state $q'$ with probability $\mu(q')$ where $\mu$ is the distribution provided by $t(q,a)$.  The goal of the agent is to select actions to maximize its expected total discounted reward
$ \E\left[\sum_{i=0}^{\infty} \gamma^i \rho_i\right] $
where $i \in \N$ (the set of natural numbers) ranges over time modeled as discrete steps, $\rho_i$ is the reward at time $i$, and the expectation is taken over the probabilistic transitions.  

We formalize the agent's plan as a \emph{stationary strategy} (commonly called a ``policy'', but we reserve that word for privacy policies).  A stationary strategy is a function $\sigma$ from the state space $\mathcal{Q}$ to the set $\mathcal{A}$ of actions (i.e., $\sigma : \mathcal{Q} \to \mathcal{A}$) such that at a state $q$ in $\mathcal{Q}$, the agent always selects to perform the action $\sigma(q)$.  
Given a strategy $\sigma$ for an MDP $m$, its expected total discounted reward is 
\[ V_m(\sigma, q) = r(q,\sigma(q)) + \gamma \sum_{q' \in \mathcal{Q}} t(q,\sigma(q))(q') * V_m(\sigma, q') \]
The agent selects one of the strategies that optimizes this equation.
We denote this set of optimal strategies as $\opt(\langle \mathcal{Q}, \mathcal{A}, t, r, \gamma\rangle)$, or when the transition system is clear from context, as $\opt(r)$.  Such strategies are sufficient to maximize the agent's expected total discounted reward despite only depending upon the MDP's current state. 

Given the strategy $\sigma$ and the actual results of the probabilistic transitions yielded by $t$, the agent exhibits an \emph{execution}.  We represent this execution as an infinite sequence $\e = [q_1, a_1, q_2, a_2, \ldots ]$ of alternating states and actions starting with a state, where $q_i$ is the $i$th state that the agent was in and $a_i$ is the $i$th action the agent took, for all $i$ in $\N$.  We say an execution $\e$ is \emph{consistent} with a strategy $\sigma$ iff $a_{i} = \sigma(q_i)$ for all $i$ in $\N$ where $a_i$ is the $i$th action in $\e$ and $q_i$ is the $i$th state in $\e$.  We call a finite prefix of an execution a \emph{behavior}.  A behavior is consistent with a strategy if it can be extended to an execution consistent with that strategy. 

Under this formalism, the auditee plays the role of the agent optimizing the MDP to plan.  We presume that each purpose may be modeled as a reward function.  That is, we assume the degree to which a purpose is satisfied may be captured by a function from states and actions to a real number.  The higher the number, the higher the degree to which that purpose is satisfied.  When the auditee wants to plan for a purpose $p$, it uses a reward function, $r^p$, such that $r^p(q,a)$ is the degree to which taking the action $a$ from state $q$ aids the purpose $p$.  We also assume that the expected total discounted reward can capture the degree to which a purpose is satisfied over time.  We say that the auditee plans \emph{for} the purpose $p$ when the auditee adopts a strategy $\sigma$ that is optimal for the MDP $\langle \mathcal{Q}, \mathcal{A}, t, r^p, \gamma\rangle$.  The appendix provides additional background information on MDPs.

\subsection{Non-redundancy}
\label{sec:non-redun}

MDPs do not require that strategies be non-redundant.  Even given that the auditee had an execution $\e$ from using a strategy $\sigma$ in $\opt(r^p)$, some actions in $\e$ might not be \emph{for} the purpose $p$.  The reason is that some actions may be redundant despite being costless.  The MDP optimization criterion behind $\opt$ prevents redundant actions from delaying the achievement of a goal as the reward associated with that goal would be further discounted making such redundant actions sub-optimal.  However, the optimization criterion is not affected by redundant actions when they appear after all actions that provide non-zero rewards.  Intuitively, the hypothetical agent planning only for the purpose in question would not perform such unneeded actions even if they have zero reward.  
Thus, to create our formalism of non-redundant MDPs (NMDPs), we replace $\opt$ with a new optimization criterion $\opt^*$ that prevents these redundant actions while maintaining the same transition structure as a standard MDP.

To account for redundant actions, we must first contrast that with doing nothing.  Thus, we introduce a distinguished action $\nothing$ that stands for doing nothing.   For all states $q$, $\nothing$ labels a transition with zero reward (i.e., $r(q,\nothing) = 0$) that is a self-loop (i.e., $t(q,\nothing)(q) = 1$).
(We could put $\nothing$ on only the subset of states that represent possible stopping points by slightly complicating our formalism.) 
Since we only allow deterministic stationary strategies and $\nothing$ only labels self-loops, this decision is irrevocable: once nothing is done, it is done forever.  
As selecting to do nothing results in only zero rewards henceforth, it may be viewed as stopping with the previously acquired total discounted reward.

Given an execution $\e$, let $\beforenothing(\e)$ denote the prefix of $\e$ before the first instance of the nothing actions.  $\beforenothing(\e)$ will be equal to $\e$ in the case where $\e$ does not contain the nothing action.

We use the idea of \emph{nothing} to make formal when one execution intuitively contain more actions than another despite both being of infinite length.  An execution $\e_1$ is a \emph{proper sub-execution} of an execution $\e_2$ if and only if $\beforenothing(\e_1)$ is a proper subsequence of $\beforenothing(\e_2)$ using the standard notion of subsequence.  Note if $\e_1$ does not contain the nothing action, it cannot be a proper sub-execution of any execution.

To compare strategies, we construct all the executions they could produce.  To do so, let a \emph{contingency} $\kappa$ be a function from $\mathcal{Q} \cross \mathcal{A} \cross \N$ to $\mathcal{Q}$ such that $\kappa(q,a,i)$ is the state that results from taking the action $a$ in the state $q$ the $i$th time.  We say that a contingency $\kappa$ is \emph{consistent} with an MDP iff $\kappa$ only picks states to which the transition function $t$ of the MDP assigns a non-zero probability to (i.e., for all $q$ in $\mathcal{Q}$, $a$ in $\mathcal{A}$, and $i$ in $\N$, $t(q,a)(\kappa(q,a,i)) > 0$).  Given an MDP $m$, let $m(q,\kappa)$ be the possibly infinite state model that results of having $\kappa$ resolve all the probabilistic choices in $m$ and having the model start in state $q$.  Let $m(q,\kappa, \sigma)$ denote the execution that results from using the strategy $\sigma$ and state $q$ in the non-probabilistic model $m(q,\kappa)$.  Henceforth, we only consider contingencies consistent with the model under discussion. 

Given two strategies $\sigma$ and $\sigma'$, we write $\sigma' \substrg \sigma$ if and only if for all contingencies $\kappa$ and states $q$, $m(q,\kappa,\sigma')$ is a proper sub-execution of or equal to $m(q,\kappa,\sigma)$, and for at least one contingency $\kappa'$ and state $q'$, $m(q',\kappa', \sigma')$ is a proper sub-execution $m(q',\kappa', \sigma)$.  Intuitively, $\sigma'$ proves that $\sigma$ produces a redundant execution under $\kappa'$ and $q'$.
We define $\opt^*(r)$ to be the subset of $\opt(r)$ holding only strategies $\sigma$ such that for no $\sigma' \in \opt(r)$ does $\sigma' \substrg \sigma$.  The following theorem, proved in the appendix, shows that non-redundant optimal strategies always exist.
 
\begin{theorem}\label{thm:opt-not-empty}
For all environment models $m$, $\opt^*(m)$ is not empty.
\end{theorem}

\subsection{Example}
\label{sec:formalism-ex}

Suppose an auditor is inspecting a hospital and comes across a physician referring a medical record to his own private practice for analysis of an X-ray as described in Section~\ref{sec:motivation}.  As physicians may only make such referrals for the purpose of treatment ($\mathsf{treat}$), the auditor may find the physician's behavior suspicious.  To investigate, the auditor may formally model the hospital using our formalism.  

The auditor would construct the NMDP $m_{\mathsf{ex}} = \langle Q_{\mathsf{ex}}, A_{\mathsf{ex}}, t_{\mathsf{ex}}, r_{\mathsf{ex}}^{\mathsf{treat}}, \gamma_{\mathsf{ex}}\rangle$ shown in Figure~\ref{fig:mot-ex}.
\begin{figure*}
\begin{center}
\input{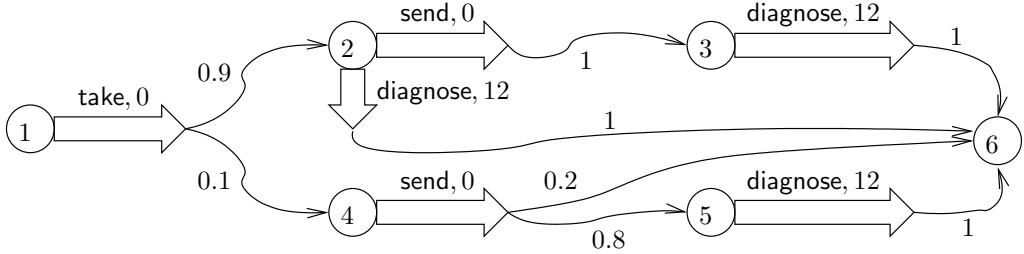}
\end{center}
\caption{The environment model $m_{\mathsf{ex}}$ that the physician used.  Circles represent states, block arrows denote possible actions, and squiggly arrows denote probabilistic outcomes.  Self-loops of zero reward under all actions, including the special action $\nothing$, are not shown.}\label{fig:mot-ex}
\end{figure*}
The figure conveys all components of the NMDP except $\gamma_{\mathsf{ex}}$.  For instance, the block arrow from the state $1$ labeled $\mathsf{take}$ and the squiggly arrows leaving it denote that after the agent performs the action $\mathsf{take}$ from state $1$, the environment will transition to the state $2$ with probability $0.9$ and to state $4$ with probability of $0.1$ (i.e., $t_{\mathsf{ex}}(1,\mathsf{take})(2) = 0.9$ and $t_{\mathsf{ex}}(1,\mathsf{take})(4) = 0.1$).  The number over the block arrow further indicates the degree to which the action satisfies the purpose of $\mathsf{treat}$.  In this instance, it shows that $r_{\mathsf{ex}}^{\mathsf{treat}}(1,\mathsf{take}) = 0$.  This transition models the physician taking an X-ray.  With probability $0.9$, he is able to make a diagnosis right away (from state $2$); with probability $0.1$, he must send the X-ray to his practice to make a diagnosis.  Similarly, the transition from state $4$ models that his practice's test has a $0.8$ success rate of making a diagnosis; with probability $0.2$, no diagnosis is ever reached.

Using the model, the auditor computes $\opt(r^{\mathsf{treat}}_{\mathsf{ex}})$, which consists of those strategies that maximizes the expected total discounted degree of satisfaction of the purpose of treatment where the expectation is over the probabilistic transitions of the model.  $\opt(r^{\mathsf{treat}}_{\mathsf{ex}})$ includes the appropriate strategy 
$\sigma_{1}$ where $\sigma_{1}(1) = \mathsf{take}$, 
$\sigma_{1}(4) = \mathsf{send}$, 
$\sigma_{1}(2) = \sigma_{1}(3) = \sigma_{1}(5) = \mathsf{diagnose}$, and
$\sigma_{1}(6) = \nothing$.
Furthermore, $\opt(r^{\mathsf{treat}}_{\mathsf{ex}})$ excludes the redundant strategy $\sigma_2$ that performs a redundant $\mathsf{send}$ where $\sigma_2$ is the same as $\sigma_1$ except for $\sigma_2(2) = \mathsf{send}$.  Performing the extra action $\mathsf{send}$ delays the reward of $12$ for achieving a diagnosis resulting in its discounted reward being $\gamma_{\mathsf{ex}}^2 * 12$ instead of $\gamma_{\mathsf{ex}} * 12$ and, thus, the strategy is not optimal.

However, $\opt(r^{\mathsf{treat}}_{\mathsf{ex}})$ does include the redundant strategy $\sigma_3$ that is the same as $\sigma_1$ except for $\sigma_3(6) = \mathsf{send}$.   $\opt(r^{\mathsf{treat}}_{\mathsf{ex}})$ includes this strategy despite the $\mathsf{send}$ actions from state $6$ being redundant since no positive rewards follow the $\mathsf{send}$ actions.  Fortunately, $\opt^*(r^{\mathsf{treat}}_{\mathsf{ex}})$ does not include $\sigma_3$ since $\sigma_1$ is both in $\opt(r^{\mathsf{treat}}_{\mathsf{ex}})$ and $\sigma_1 \substrg \sigma_3$.  To see that $\sigma_1 \substrg \sigma_3$ note that for every contingency $\kappa$ and state $q$, the $m_{\mathsf{ex}}(q,\kappa, \sigma_1)$ has the form $\be$ followed by an finite sequence of nothing actions (interleaved with the state $6$) for some finite prefix $\be$.  For the same $\kappa$, $m_{\mathsf{ex}}(q,\kappa, \sigma_3)$ has the form $\be$ followed by an infinite sequence of $\mathsf{send}$ actions (interleaved with the state $6$) for the same $\be$.  Thus, $m_{\mathsf{ex}}(q,\kappa, \sigma_1)$ is a proper sub-execution of $m_{\mathsf{ex}}(q,\kappa, \sigma_3)$.

\section{Auditing}
\label{sec:audit}

In the above example, the auditor constructed a model of the environment in which the auditee operates.  The auditor must use the model to determine if the auditee obeyed the policy.  We first discuss this process for auditing restrictive policy rules and revisit the above example.  Then, we discuss the process for prohibitive policy rules.  In the next section, we provide an auditing algorithm that automates comparing the auditee's behavior, as recorded in a log, to the set of allowed behaviors.

\subsection{Auditing Restrictive Rules}
\label{sec:audit-only}

Suppose that an auditor would like to determine whether an auditee performed some logged actions \emph{only for} the purpose $p$.  The auditor can compare the logged behavior to the behavior that a hypothetical agent would perform when planning for the purpose $p$.  In particular, the hypothetical agent selects a strategy from $\opt^*(\langle \mathcal{Q}, \mathcal{A}, t, r^p, \gamma\rangle)$ where $\mathcal{Q}$, $\mathcal{A}$, and $t$ models the environment of the auditee; $r^p$ is a reward function modeling the degree to which the purpose $p$ is satisfied; and $\gamma$ is an appropriately selected discounting factor.  If the logged behavior of the auditee would never have been performed by the hypothetical agent, then the auditor knows that the auditee violated the policy.

In particular, the auditor must consider all the possible behaviors the hypothetical agent could have performed.  For a model $m$, let $\behv^*(r^p)$ represent this set where a finite prefix $\be$ of an execution is in $\behv^*(r^p)$ if and only if there exists a strategy $\sigma$ in $\opt^*(r^p)$, a contingency $\kappa$, and a state $q$ such that $\be$ is a subsequence of $m(q, \kappa, \sigma)$.

The auditor must compare $\behv^*(r^p)$ to the set of all behaviors that could have caused the auditor to observe the log that he did.  
We presume that the log $\ell$ was created by a process $\mathsf{log}$ that records features of the current behavior.  That is, $\mathsf{log}{:}\, B \to L$ where $B$ is the set of behaviors and $L$ the set of logs, and $\ell = \mathsf{log}(\be)$ where $\be$ is the prefix of the actual execution of the environment available at the time of auditing.
The auditor must consider all the behaviors in $\mathsf{log}^{-1}(\ell)$ as possible where $\mathsf{log}^{-1}$ is the inverse of the logging function.
In the best case for the auditor, the log records the whole prefix $\be$ of the execution that transpired until the time of auditing, in which case $\mathsf{log}^{-1}(\ell) = \{\ell\}$.

If $\mathsf{log}^{-1}(\ell) \cap \behv^*(r^p)$ is empty, then the auditor may conclude that the auditee did not plan for the purpose $p$, and, thus, violated the rule that auditee must only perform the actions recorded in $\ell$ for the purpose $p$; otherwise, the auditor must consider it possible that the auditee planned for the purpose $p$.

If $\mathsf{log}^{-1}(\ell) \subseteq \behv^*(r^p)$, the auditor might be tempted to conclude that the auditee surely obeyed the policy rule.  However, as illustrated in the second example below, this is not necessarily true.  The problem is that $\mathsf{log}^{-1}(\ell)$ might have a non-empty intersection with $\behv^*(r^{p'})$ for some other purpose $p'$.  In this case, the auditee might have been actually planning for the purpose $p'$ instead of $p$.  Indeed, given the likelihood of such other purposes for non-trivial scenarios, we consider proving compliance practically impossible.
However, this incapability is of little consequence: $\mathsf{log}^{-1}(\ell) \subseteq \behv^*(r^p)$ does imply that the auditee is behaving as though he is obeying the policy.  That is, in the worse case, the auditee is still doing the right things even if for the wrong reasons.

\subsection{Example}
\label{sec:audit-only-ex}

Below we revisit the example of Section~\ref{sec:formalism-ex}.  We consider two cases.  In the first, the auditor shows that the physician violated the policy.  In the second, auditing is inconclusive.

\paragraph{Violation Found}
Suppose after constructing the model as above in Section~\ref{sec:formalism-ex}, the auditor maps the actions recorded in the access log $\ell_1$ to the actions of the model $m_{\mathsf{ex}}$, and finds $\mathsf{log}^{-1}(\ell_1)$ holds only a single behavior: $\be_{1} = [1, \mathsf{take}, 2, \mathsf{send},  3, \mathsf{diagnose}, 6, \nothing, 6]$.  
Next, using $\opt^*(r^{\mathsf{treat}}_{\mathsf{ex}})$, as computed above, the auditor constructs the set $\behv^*(r^{\mathsf{treat}}_{\mathsf{ex}})$ of all behaviors an agent planning for treatment might exhibit.  The auditor would find that $\be_{1}$ is not in $\behv^*(r^{\mathsf{treat}}_{\mathsf{ex}})$.  

To see this, note that every execution $\e_1$ that has $\be_1$ as a prefix is generated from a strategy $\sigma$ such that $\sigma(2) = \mathsf{send}$.  The strategy $\sigma_2$ from Section~\ref{sec:formalism-ex} is one such strategy.  None of these strategies are members of $\opt(r^{\mathsf{treat}}_{\mathsf{ex}})$ for the same reason as $\sigma_2$ is not a member.  Thus, $\be_1$ cannot be in $\behv^*(r^{\mathsf{treat}}_{\mathsf{ex}})$.  As $\mathsf{log}^{-1}(\ell) \cap \behv^*(r^{\mathsf{treat}}_{\mathsf{ex}})$ is empty, the audit reveals that the physician violated the policy.

\paragraph{Inconclusive}
Now suppose that the auditor sees a different log $\ell_2$ such that $\mathsf{log}^{-1}(\ell_2) = \{ \be_2 \}$ where $\be_{2} = [1, \mathsf{take}, 4, \mathsf{send}, 5, \mathsf{diagnose}, 6, \nothing, 6]$.  In this case, our formalism would not find a violation since $\be_{2}$ is in $\behv^*(r^{\mathsf{treat}}_{\mathsf{ex}})$.  In particular,  the strategy $\sigma_{1}$ from above produces the behavior $\be_2$ under the contingency that selects the bottom probabilistic transition from state $1$ to state $4$ under the action $\mathsf{take}$.

Nevertheless, the auditor cannot be sure that the physician obeyed the policy.  
For example, consider the NMDP $m'_{\mathsf{ex}}$ that is $m_{\mathsf{ex}}$ altered to use the reward function $r^\mathsf{profit}_{\mathsf{ex}}$ instead of $r^\mathsf{treat}_{\mathsf{ex}}$.  $r^{\mathsf{profit}}_{\mathsf{ex}}$ assigns a reward of zero to all transitions except for the $\mathsf{send}$ actions from states $2$ and $4$, to which it assigns a reward of $9$.  $\sigma_{1}$ is in $\opt^*(r^{\mathsf{profit}}_{\mathsf{ex}})$ meaning that not only the same actions (those in $\be_{2}$), but even the exact same strategy can be either for the allowed purpose $\mathsf{treat}$ or the disallowed purpose $\mathsf{profit}$.  Thus, if the physician did refer the record to his practice for profit, he cannot be caught as he has tenable deniability of his ulterior motive of profit.

\subsection{Auditing Prohibitive Rules}
\label{sec:audit-not}

In the above example, the auditor was enforcing the rule that the physician's actions be \emph{only for} treatment.  Now, consider auditing to enforce the rule the that physician's actions are \emph{not for} personal profit.  After seeing the log $\ell$, the auditor could check whether $\mathsf{log}^{-1}(\ell) \cap \behv^*(r^{\mathsf{profit}}_{\mathsf{ex}})$ is empty.  If so, then the auditor knows that the policy was obeyed.  If not, then the auditor cannot prove nor disprove a violation.
In the above example, just as the auditor is unsure whether the actions were \emph{for} the required purpose of treatment, the auditor is unsure whether the actions are \emph{not for} the prohibited purpose of profit.  

An auditor might decide to investigate some of the cases where $\mathsf{log}^{-1}(\ell) \cap \behv^*(r^{\mathsf{profit}}_{\mathsf{ex}})$ is not empty.
In this case, the auditor could limit his attention to only those possible violations of a prohibitive rule that cannot be explained away by some allowed purpose.  For example, in the inconclusive example above, the physician's actions can be explained with the allowed purpose of treatment.  As the physician has tenable deniability, it is unlikely that investigating his actions would be a productive use of the auditor's time.  Thus, the auditor should limit his attention to those logs $\ell$ such that both $\mathsf{log}^{-1}(\ell) \cap \behv^*(r^{\mathsf{profit}}_{\mathsf{ex}})$ is non-empty and  $\mathsf{log}^{-1}(\ell) \cap \behv^*(r^{\mathsf{treat}}_{\mathsf{ex}})$ is empty.

A similar additional check using disallowed purposes could be applied to enforcing restrictive rules.  However, for restrictive rules, this check would identify cases where the auditee's behavior could have been either for the allowed purpose or a disallowed purpose.  Thus, it would serve to find additional cases to investigate and increase the auditor's workload rather than reduce it.  Furthermore, the auditee would have tenable deniability for these possible ulterior motives, making these investigations a poor use of the auditor's time.

\section{Auditing Algorithm}
\label{sec:algo}

We would like to automate the auditing process described above.  To this end, we present in Figure~\ref{fig:algo} an algorithm $\mathsc{Audit}$ that aids the auditor in comparing the log to the set of allowed behaviors.
As we are not interested in the details of the logging process and would like to focus on the planning aspects of our semantics, we limit our attention to the case where $\mathsf{log}(\be) = \be$.  As proved below (Theorem~\ref{thm:algo-correct}), $\mathsc{Audit}(m,\be)$ returns true if and only if $\mathsf{log}^{-1}(\be) \cap \behv^*(m)$ is empty.  
In the case of a restrictive rule, the auditor may conclude that the policy was violated when $\mathsc{Audit}$ returns true.  In case of a prohibitive rule, the auditor may conclude the policy was obeyed when $\mathsc{Audit}$ returns true.

\newcommand{\s}{\ \ \ }
\newcommand{\ls}{\ }
\begin{figure}
\mbox{}$\mathsc{Audit}(\langle \mathcal{Q}, \mathcal{A}, t, r, \gamma\rangle, [q_0, a_1, q_1, \ldots a_n, q_n])$:\\
\mbox{} 01 \ls $V^*_m := \mathsf{solveMDP}(\langle \mathcal{Q}, \mathcal{A}, t, r, \gamma\rangle)$\\
\mbox{} 02 \ls for($i := 0$; $i < n$; $i\mathtt{++}$):\\
\mbox{} 03 \ls\s if($a_{i+1} \neq \nothing$):\\
\mbox{} 04 \ls\s\s if($r[q_i][a_{i+1}] + \gamma \sum_{j = 0}^{|\mathcal{Q}|} t[q_i][a_{i+1}][j] * V^*_m[j] \leq 0$):\\
\mbox{} 05 \ls\s\s\s return true\\
\mbox{} 06 \ls $r^* := 0$\\
\mbox{} 07 \ls for($j := 0$; $j < |\mathcal{Q}|$; $j\mathtt{++}$):\\
\mbox{} 08 \ls\s for($k := 0$; $k < |\mathcal{A}|$; $k\mathtt{++}$):\\
\mbox{} 09 \ls\s\s $r'[j][k] := r[j][k]$\\
\mbox{} 10 \ls\s\s if($r^* < \mathsf{absoluteValue}(r[j][k])$:\\
\mbox{} 11 \ls\s\s\s $r^* := \mathsf{absoluteValue}(r[j][k])$\\
\mbox{} 12 \ls $\omega := 2*r^*/(1-\gamma) + 1$\\
\mbox{} 13 \ls for($i := 0$; $i < n$; $i\mathtt{++}$):\\
\mbox{} 14 \ls\s for($k := 0$; $k < |\mathcal{A}|$; $k\mathtt{++}$):\\
\mbox{} 15 \ls\s\s if($k \neq a_{i+1}$):\\
\mbox{} 16 \ls\s\s\s $r'[q_i][k] := -\omega$\\
\mbox{} 17 \ls $m' := \langle\mathcal{Q}, \mathcal{A}, t, r', \gamma\rangle$\\
\mbox{} 18 \ls $V^*_{m'} := \mathsf{solveMDP}(\langle \mathcal{Q}, \mathcal{A}, t, r', \gamma\rangle)$\\
\mbox{} 19 \ls for($j := 0$; $j < |\mathcal{Q}|$; $j\mathtt{++}$):\\
\mbox{} 20 \ls\s if($V^*_m[j] = V^*_{m'}[j]$):\\
\mbox{} 21 \ls\s\s return false\\
\mbox{} 22 \ls return true
\caption{The algorithm $\mathsc{Audit}$.  $\mathsf{solveMDP}$ may be any MDP solving algorithm.  The algorithm assumes functions are represented as arrays and states and actions are represented as indexes into these arrays.}
\label{fig:algo}
\end{figure}

$\mathsc{Audit}$ operates in two steps.  The first checks to make sure that the behavior $\be$ is not inherently redundant (lines 01--05).  If it is, then $\mathsf{log}^{-1}(\be) \cap \behv^*(m)$ will be empty and the algorithm returns true.  $\mathsc{Audit}$ checks $\be$ by comparing the actions taken in each state to doing nothing.  If the expected total discounted reward for doing nothing in a state $q$ is higher than that for doing the action $a$ in $q$, then $a$ introduces redundancy into any strategy $\sigma$ such that $\sigma(q) = a$.  Thus, if $b = [\ldots, q, a, \ldots]$, we may conclude that $\mathsf{log}^{-1}(\be) \cap \behv^*(m)$ is empty.

The second step compares the optimal values of two MDPs.  One of the them is the NMDP $m$ treated as an MDP, which is already optimized during the first step.  The other $m'$ is constructed from $m$ (lines 07--17) so that only the actions in the log $\be$ are selected during optimization.  If the expected total discounted reward of each of these MDPs is unequal, then $\mathsf{log}^{-1}(\be) \cap \behv^*(m)$ is empty.  

Below we formalize these ideas.  Lemma~\ref{lem:two-steps} justifies our two step approach while Lemmas~\ref{lem:useless-check} and~\ref{lem:fix-value-check} justify how we perform the first and second step, respectively.  They allow us to conclude the correctness of our algorithm in Theorem~\ref{thm:algo-correct}.  We defer proofs and additional propositions to the appendix.

\subsection{Useless States and the Two Steps}

We say an action is \emph{useless} at a state if taking it would always lead to redundancy.  Formally, let the set $U_m$ be the subset of $\mathcal{Q} \cross \mathcal{A}$ such that $\langle q, a\rangle$ is in $U_m$ if and only if  $a \neq \nothing$ and for all strategies $\sigma$, $Q_m(\sigma,q,a) \leq 0$ where $Q_m(\sigma, q, a) = r(q,a) + \gamma \sum_{q'} t(q,a)(q') * V_m(\sigma, q')$.

We call $\langle q, a\rangle$ in set $U_m$ \emph{useless} since any strategy $\sigma$ such that $\sigma(q) = a$ could be replaced by a strategy $\sigma'$ that is the same as $\sigma$ except for having $\sigma'(q) = \nothing$ without lowering the expected total discounted reward.  
To make this formal, let $U(\sigma)$ be a strategy such that $U(\sigma)(q) = \nothing$ if $\langle q, \sigma(q)\rangle \in U$ and $U(\sigma)(q) = \sigma(q)$ otherwise.  The following justifies calling these pairs \emph{useless}: for all $\sigma$ and $q$, $V_m(\sigma, q) \leq V_m(U_m(\sigma), q)$ (Proposition~\ref{prp:useless}).

We are also interested in the set $\strg(\be)$ of strategies that could have resulted in the behavior $b$: $\strg(\be) = \set{\sigma \in \mathcal{Q} \to \mathcal{A}}{\forall i < n . a_{i+1} = \sigma(q_{i})}$ where $\be = [q_0, a_1, q_1, a_2, \ldots, a_n, q_n]$.  

\begin{lemma}\label{lem:two-steps}
For all environment models $m$ and all behaviors $\be = [q_0, a_1, q_1, \ldots, a_n, q_n]$,
$\mathsf{log}^{-1}(\be) \cap \behv^*(m)$ is empty
if and only if
(1) there exists $i$ such that $0 \leq i < n$ and $\langle q_i, a_{i+1}\rangle \in U_m$ or
 (2) $\strg(b) \cap \opt(m)$ is empty,
\end{lemma}

Thus, checking whether $\mathsf{log}^{-1}(\be) \cap \behv^*(m)$ is empty has been reduced to checking the two conditions (1) and (2).  We explain how to check each of these in the next two sections.
 
\subsection{Step 1: Inherent Redundancy}

Rather than construct $U_m$ explicitly, we use the following lemma to check condition (1).
The lemma uses the definition 
$Q^*_m(q,a) = r(q,a) + \gamma \sum_{q'} t(q,a)(q') * V^*_m(q')$ where
$V^*_m(q) = \max_{\sigma} V_m(\sigma, a)$.

\begin{lemma}\label{lem:useless-check}
For all environment models $m$, states $q$, and actions $a$, $\langle q, a\rangle$ is in $U_m$ if and only if $a \neq \nothing$ and $Q^*_m(q,a) \leq 0$.
\end{lemma}

\subsection{Step 2: Checking Optimality}

To check (2), we construct a model $m'$ from $m$ that limits the optimization to selecting a strategy that can cause the observed behavior $\be$.  To do so, we adjust the reward function of $m$ so that the actions taken in $b$ are always taken by the optimal strategies of $m'$.  That is, if $b = [q_0,a_1,q_1,\ldots, a_n,q_n]$, then for each $q_i$ and $a_{i+1}$, we replace the reward for taking an action $a'$ other than $a_{i+1}$ from the state $q$ with a negative reward $-\omega$ that is so low as to assure that the action $a'$ would not be used by any optimal strategy.   
We use $\omega > 2r^*/(1-\gamma)$ where $r^*$ is the reward with the largest magnitude appearing in $m$ since the total discounted reward is bounded from below by $-r^*/(1-\gamma)$ and from above by $r^*/(1-\gamma)$ (recall that $\sum_{i=0}^{\infty} \gamma^i r^* = r^*/(1-\gamma)$).

We formally define $m'$ to be $\fix(m,\be)$ where $\fix(m,[]) = m$ and 
\[ \fix(\langle \mathcal{Q}, \mathcal{A}, t, r, \gamma\rangle, [q_0,a_1,q_1,\ldots, a_n,q_n]) =  \fix(\langle \mathcal{Q}, \mathcal{A}, t, r', \gamma\rangle, [q_1,\ldots,a_n,q_n]) \]
where $r'(q_0,a) = -\omega$ for all $a \neq a_1$ and $r'(q_0, a_1) = r(q_0, a_1)$.
The construction $\fix$ has the following useful property:  $\strg(b) \cap \opt(m)$ is empty if and only if $\opt(\fix(m,b)) \cap \opt(m)$ is empty (Proposition~\ref{prp:fix-works}).
This property is useful since testing whether $\opt(m) \cap \opt(\fix(m,\be))$ is empty may be reduced to simply comparing their optimal values: $\opt(m) \cap \opt(\fix(m,\be))$ is empty if and only if for all states $q$, $\max_{\sigma} V_{\fix(m,\be)}(\sigma,q) \neq \max_{\sigma} V_m(\sigma, q)$ (Proposition~\ref{prp:value-check}).  
Fortunately, algorithms exist for finding the optimal value of MDPs (see, e.g.,~\cite{rn03artificial}).

These two propositions combine to yield the next lemma, which justifies how we conduct testing for the second condition of Lemma~\ref{lem:two-steps} in the second step of $\mathsc{Audit}$. 
\begin{lemma}\label{lem:fix-value-check}
For all environment models $m$ and behaviors $b$, $\strg(b) \cap \opt(m)$ is empty if and only if for all $q$, $\max_{\sigma} V_{\fix(m,\be)}(\sigma,q) \neq \max_{\sigma} V_m(\sigma, q)$.
\end{lemma}

These lemmas combine with reasoning about the actual code of the program to yield its correctness.
\begin{theorem}\label{thm:algo-correct}
For all environment models $m$ and behaviors $b$, $\mathsc{audit}(m,\be)$ returns true if and only if $\mathsf{log}^{-1}(\be) \cap \behv^*(m)$ is empty.
\end{theorem}

The running time of the algorithm is dominated by the two MDP optimizations.  These may be done exactly by reducing the optimization to a system of linear equations~\cite{d63probabilistic}.  Such systems may be solved in polynomial time~\cite{k79polynomial,k84new}.  However, in practice, large systems are often difficult to solve.  Fortunately, a large number of algorithms for making iterative approximations exist whose run time depends on the quality of the approximation.  (See~\cite{ldk95complexity} for a discussion.)

\section{Applying our Formalism to Past Methods}
\label{sec:compare}

Past methods of enforcing purpose requirements have not provided methods of assigning purposes to sequences of actions.  
Rather, they presume that the auditor (or someone else) already has a method of determining which behaviors are for a purpose.  
In essence, these methods presuppose that the auditor already has the set of allowed behaviors $\behv^*(r^p)$ for the purpose $p$ that he is enforcing.  These methods differ in their intensional representations of the set $\behv^*(r^p)$.  Thus, some may represent a given set exactly while others may only be able to approximate it.  These differences mainly arise from the different mechanisms they use to ensure that the auditee only exhibits behaviors from $\behv^*(r^p)$.  We use our semantics to study how reasonable these approximations are.

Byun et al.\ use role-based access control~\cite{s96role} to consider purposes~\cite{bbl05purpose,bl08purpose,nblbkkt10privacy}.  They associate purposes with sensitive resources and with roles, and their method only grants the user access to the resource when the purpose of the user's role matches the resource's purpose.  The method does not, however, explain how to determine which purposes to associate with which roles.  Furthermore, a user in a role can perform actions that do not fit the purposes associated with his role allowing him to use the resource for a purpose other than the intended one.  Thus, their method is only capable of enforcing policies when there exists some subset $A$ of the set of actions $\mathcal{A}$ such that $\behv^*(r^p)$ is equal to the set of all interleavings of $A$ with $\mathcal{Q}$ of finite but unbounded length (i.e., $\behv^*(r^p) = (\mathcal{Q} \cross A)^* \cons \mathcal{Q}$ where $\cons$ is append raised to work over sets in the standard pairwise manner).  The subset $A$ corresponds to those actions that use a resource with the same purpose as the auditee's role.
Despite these limitations, their method can implement the run-time enforcement used at some organizations, such as a hospital that allows physicians access to any record to avoid denying access in time-critical emergencies.  However, it does not allow for the fine-grain distinctions used during post-hoc auditing done at some hospitals to ensure that physicians do not abuse their privileges.

Al-Fedaghi uses the work of Byun et al.\ as a starting point but concludes that rather than associating purposes with roles, one should associate purposes with sequences of actions~\cite{a07beyond}.  Influenced by Al-Fedaghi, Jafari et al.\ adopt a similar position calling these sequences \emph{workflows}~\cite{jss09enforcing}.  The set of workflows allowed for a purpose $p$ corresponds to $\behv^*(r^p)$.  They do not provide a formal method of determining which workflows belong in the allowed set.  They do not consider probabilistic transitions and the intuition they supply suggests that they would only include workflows that successfully achieves or improves the purpose.  Thus, our approach appears more lenient by including some behaviors that fail to improve the purpose.

Others have adopted a hybrid approach allowing for the roles of an auditee to change based on the state of the system~\cite{pgy08dynamic,kwb11conditional}.  These changes effectively allow role-based access control to simulate the workflow methods to be just as expressive while introducing a level of indirection inhabited by dynamic roles.

Agrawal et al.\ use a \emph{query intrusion model} to enforce purpose requirements that operates in a manner similar to intrusion detection~\cite{aksx02hippocratic}.  Their method flags a request for access as a possible violation if the request claims to be for a purpose despite being dissimilar to previous requests for the same purpose.  To avoid false positives, the set of allowed behaviors $\behv^*(r^p)$ would have to be small or have a pattern that the query intrusion model could recognize.

Jif is a language extension to Java designed to enforce requirements on the flows of information in a program~\cite{cmvz09jif}.  Hayati and Abadi explain how to reduce purpose requirements to information flow properties that Jif can enforce~\cite{ha05language}.  Their method requires that inputs are labeled with the purposes for which the policy allows the program to use them and that each unit of code be labeled with the purposes for which that code operates.  If information can flow from an input statement labeled with one purpose to code labeled for a different purpose, their method produces a compile-time type error.  (For simplicity, we ignore their use of sub-typing to model sub-purposes.)  In essence, their method enforces the rule \emph{if information $i$ flows to code $c$, then $i$ and $c$ must be labeled with the same purpose}.  The interesting case is when the code $c$ uses the information $i$ to perform some observable action $a_{c,i}$, such as producing output.  Under our semantics, we treat the program as the auditee and view the policy as limiting these actions.
By directly labeling code, their method does not consider the contexts in which these actions occur.  Rather the action $a_{c,i}$ is aways either allowed or not based on the purpose labels of $c$ and $i$.  By not considering context, their method is subject to the same limitations as the method of Byun et al.\ with the subset $A$ being equal to the set of all actions $a_{c,i}$ such that $c$ and $i$ have the same label.
However, using more advanced type systems (e.g., typestate~\cite{sy86typestate}), they might be able extend their method to consider the context in which code is executed and increase the method's expressiveness.

\section{Multiple Purposes}
\label{sec:multi-purposes}

So far, our formalism allows our hypothetical agent to consider only a single purpose.  However, auditees may perform an action for more than one purpose.  In many cases, the auditor may simply ignore any action that is not governed by the privacy policy and not relevant to the plans the auditee is employing that uses governed actions.  

In the physician example above, the physician already implicitly considered many other purposes before even seeing this current patient.  For example, the physician presumably performed many actions not mentioned in the model in between taking the X-ray, sending it, and making a diagnosis, such as going on a coffee break.  As these actions are not governed by the privacy policy and neither improves nor degrades the diagnosis even indirectly, the auditor may safely ignore them.  Thus, our semantics can handle multiple purposes in this limited fashion.

However, in other cases, the interactions between purposes become important.
Below we discuss two complementary ways that an auditee can consider multiple purposes that produce interactions.  In the first, the auditee considers one purpose after another.  In the second, the auditee attempts to optimize for multiple purposes simultaneously.  We find that our semantics may easily be extended to handle the first, but difficulties arise for the second.  We end the section by considering what features a formalism would need to handle simultaneous consideration of purposes and the challenges they raise for auditing.

\subsection{Sequential Consideration}
\label{sec:sequential}

Yahoo!'s privacy policy states that they will not contact children for the purpose of marketing~\cite{yahoo-info-policy}.  Suppose Yahoo!\ decides to change the name of \url{games.yahoo.com} to \url{fun.yahoo.com} because they believe the new name will be easier to market.  They notify users of \url{games.yahoo.com}, including children, of the upcoming change so that they may update their bookmarks.

In this example, the decision to change names, made for marketing, causes Yahoo!\ to contact children.  However, we do not feel this is a violation of Yahoo!'s privacy policy.  A decision made for marketing altered the expected future of Yahoo!\ in such a way that customer service would suffer if Yahoo!\ did nothing.  Thus, to maintain good customer service, Yahoo!\ made the decision to notify users without further consideration of marketing.  Since Yahoo!\ did not consider the purpose of marketing while making this decision, contacting the children was not \emph{for} marketing despite Yahoo!\ considering the implications of changing the name for marketing while making its decision to contact children.  

Bratman describes such planning in his work formalizing \emph{intentions}~\cite{b87intention}.  He views it as a sequence of planning steps in which the intention to act (e.g., to change the name) at one step may affect the plans formed at later steps.  In particular, each step of planning starts with a model of the environment that is refined by the intentions formed by each of the previous planning steps.  The step then creates a plan for a purpose that further refines the model with new intentions resulting from this plan.  Thus, a purpose of a previous step may affect the plan formed in a later step for a different purpose by constraining the choices available at the later step of planning.  We adopt the stance that an action selected at a step is \emph{for} the purpose optimized at that step but not other previous purposes affecting the step.

\subsection{Simultaneous Consideration}
\label{sec:simultaneous}

At other times, an auditee might consider more than one purpose in the same step.  For example, the physician may have to both provide quality treatment and respect the patient's financial concerns.  In this case, the physician may not be able to simultaneously provide the highest quality care at the lowest price.  The two competing concerns must be balanced and the result may not maximize the satisfaction of either of them.  

The traditional way of modeling the simultaneous optimization of multiple rewards is to combine them into a single reward using a weighted average over the rewards.  Each reward would be weighted by how important it is to the auditee performing the optimization.  This amalgamation of the various purpose rewards makes it difficult to determine for which purpose various actions are selected.

One possibility is to analyze the situation using counterfactual reasoning (see, e.g.,~\cite{m74cement}).  For example, given that the auditee performed an action $a$ while optimizing a combination of purposes $p_1$ and $p_2$, the auditor could ask if the auditee would have still performed the action $a$ even if the auditee had not considered the purpose $p_1$ and had only optimized the purpose $p_2$.  If not, than the auditor could determine that the action was for $p_1$.  However, as the next example shows, such reasoning is not sufficient to determine the purposes of the actions.

To show the generality of purposes, we consider an example involving travel reimbursement.  Consider a Philadelphian who needs to go to New York City for a business meeting with his employer and is invited to give a lecture at a conference in Washington, D.C., with his travel expenses reimbursed by the conference.  He could drive to either New York or Washington (modeled as the actions $\mathsf{driveNY}$ and $\mathsf{driveDC}$, respectively).  However, due to time constraints he cannot drive to both of them.  To attend both events, he needs to fly to both (modeled as actions $\mathsf{flyNY}$ and $\mathsf{flyDC}$).  As flying is more expensive, both driving actions receives a higher reward than flying ($2$ instead of $1$), but flying is better than not going ($0$).  Figure~\ref{fig:drive-fly} models the traveler's environment.
\begin{figure}
\begin{center}
\input{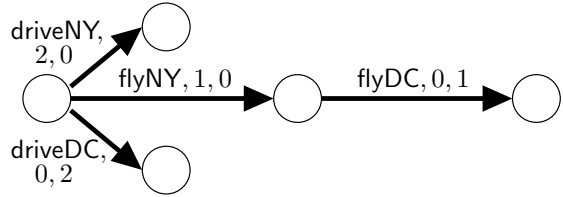}
\end{center}
\caption{Model of a traveler deciding whether to fly or drive.  Since every transition is deterministic, we represent each as a single arrow.  Each is labeled with the action name, the rewards for business and the rewards for lecturing in that order.  Self-loops of zero reward are not shown including all those labeled with the nothing action $\nothing$.}
\label{fig:drive-fly}
\end{figure}

Given these constraints, he decides to fly to both only to find auditors at both events scrutinizing his decision.
For example, an auditor working for the conference could find that his flight to Washington was not for the lecture since the traveler would have driven had it not been for work.  If the conference's policy requires that reimbursed flights are \emph{only for} the lecture, the auditor might deny reimbursement.  However, the employer seems even less likely to reimburse the traveler for his flight to Washington since the flight is redundant for getting to New York.

However, under the semantics discussed above, each flight would be for both purposes since only when the traveler considers both does he decide to take either flight.  While having the conference reimburse the traveler for his flight to Washington seems reasonable, the idea that they should also reimburse him for his flight to New York appears counterintuitive.

Our approach of sequential planning also cannot explain this example.  To plan sequentially, the traveler must consider one of the two events first.  If, for example, he considers New York first, he will decided to drive to New York and then decline the invitation to Washington.  Only by considering both events at once, does he decide to fly.  

We believe resolving this conflict requires extending our semantics to consider requirements that an action be \emph{for} a purpose (as opposed to \emph{not for} or \emph{only for}).  Furthermore, we believe that the optimization of combinations of purposes does not accurately model human planning with multiple purposes.  Intuitively, the traveler selects $\mathsf{flyDC}$ not \emph{for} work but also not \emph{only for} the conference.  Rather $\mathsf{flyDC}$ seems be \emph{for} the conference under the constraint that it must not prevent the traveler from attending the meeting.  In the next section, we consider the possibility of modeling human planning more accurately.

\subsection{Modeling Human Planning}
\label{sec:human-planning}

While MDPs are useful for automated planning, they are not specialized for modeling planning by humans, leading to the search for more tailored models~\cite{s55behavioral,gs02bounded}.  Simon proposed to model humans as having \emph{bounded rationality} to account for their limitations and their lack of information~\cite{s55behavioral}.  Work on formalizing bounded rationality has resulted in a variety of planning principles ranging from the systematic (e.g., Simon's \emph{satisficing}) to the heuristic (e.g.,~\cite{g02adaptive}).
However, ``[a] comprehensive, coherent theory of bounded rationality is not available''~\cite[p14]{s02what} and there still is ``a significant amount of unpredictability in how an animal or a human being will undertake to solve a problem'' such as planning~\cite[p40]{dkp96cognitive}.

We view creating semantics more closely tied to human planning interesting future work.  However, modeling human planning may prove complex enough to justify accepting the imperfections of semantics such as ours or even heuristic based approaches for finding violations such as the query intrusion model discussed above~\cite{aksx02hippocratic}.

Despite these difficulties, one could look for discrepancies between a semantics of purpose requirements and experimental results on planning.  In this manner one could judge how closely a semantics approximates human planning in the ways relative to purpose requirements.

In particular, our semantics appears to hold human auditees to too high of a standard: they are unlikely to always be able to pick the optimal strategy for a purpose.  When enforcing a restrictive rule, this strictness could result in the auditor investigating some auditees who honestly planned for the only allowed purpose, but failed to find the optimal policy.  While such investigations would be false positives, they do have the pleasing side-effect of highlighting areas in which an auditee could improve his planning.  

In the case of enforcing prohibitive rules, this strictness could cause the auditor to miss some violations that do not optimize the prohibited purpose, but, nevertheless, are for the purpose.  The additional checks proposed at the end of Section~\ref{sec:audit-not} could be useful for detecting these violations: if the auditee's actions are not consistent with a strategy that optimizes any of the allowed purposes but does improve to some degree the prohibited purpose, the actions may warrant extra scrutiny.

While our semantics is limited by our understanding of human planning, it still reveals concepts crucial to the meaning of \emph{purpose}.  Ideas such as planning and non-redundancy will guide future investigations on the topic.

\section{Related Work}
\label{sec:rel-work}

We have already covered the most closely related work in Section~\ref{sec:compare}.  Below we discuss work on related problems and work on purpose from other fields.

\paragraph{Minimal Disclosure}
The works most similar to ours in approach have been on \emph{minimal disclosure}, which requires that the amount of information used in granting a request for access should be as little as possible while still achieving the purpose behind the request.  Massacci, Mylopoulos, and Zannone define minimal disclosure for Hippocratic databases~\cite{mmz06hierarchical}.  
Barth, Mitchell, Datta, and Sundaram study minimal disclosure in the context of workflows~\cite{bmds07privacy}.  They model a workflow as meeting a utility goal if it satisfies a temporal logic formula.
 Minimizing the amount of information disclosed is similar to an agent maximizing his reward and thereby not performing actions that have costs but no benefits.  However, in addition to having different research goals, we consider several factors that these works do not, including quantitative purposes that are satisfied to varying degrees and probabilistic behavior resulting in actions being for a purpose despite the purpose not being achieved.

\paragraph{Expressing Privacy Policies with Purpose}
Work on understanding the components of privacy policies has shown that \emph{purpose} is a common component of privacy rules (e.g.,~\cite{ba05analyzing,ba08analyzing}).  Some languages for specifying access-control policies allow the purpose of an action to partially determine if access is granted~\cite{epal,c02web,bkkf05usable,bkk06empirical}.  However, these languages do not give a formal semantics to the purposes.  Instead they rely upon the system using the policy to determine whether an action is for a purpose or not.

\paragraph{Philosophical Foundations}
Taylor provides a detailed explanation of the importance of planning to the meaning of \emph{purpose}, but does not provide any formalism~\cite{t66action}.

The sense in which the word ``purpose'' is used in privacy policies is also related to the ideas of \emph{desire}, \emph{motivation}, and \emph{intention} discussed in works of philosophy (e.g.,~\cite{a57intention}).  
The most closely related to our work is that of Bratman's on intentions from which we get our model of sequential planning~\cite{b87intention}.  In his work, an \emph{intention} is an action an agent plans to take where the plan is formed while attempting to maximize the satisfaction of the agent's \emph{desires}; Bratman's \emph{desires} correspond to our \emph{purposes}.  
Roy formalized Bratman's work using logics and game theory~\cite{r08thinking}.  However, these works are concerned with when an action is rational rather than determining the purposes behind the action. 

We borrow the notion of \emph{non-redundancy} from Mackie's work on formalizing \emph{causality} using counterfactual reasoning~\cite{m74cement}.  In particular, Mackie defines a \emph{cause} to be a non-redundant part of a sufficient explanation of an effect.  
Roughly speaking, we replace the causes with actions and the effect with a purpose.  The extension to our semantics proposed in Section~\ref{sec:simultaneous}, may be seen as another instance of non-redundancy.  This time, we replace the causes with purposes and the effect with an action.  This suggests that for an action to be for a purpose, we expect both that the action was non-redundant for improving that purpose and that the purpose was non-redundant in motivating the action.  That is, we expect planning to be parsimonious.

\paragraph{Planning}
Psychological studies have produced models of human thought (e.g.,~\cite{abbdl04integrated}). 
However, these are too low-level and incomplete for our needs~\cite{dkp96cognitive}.
The GOMS formalism provides a higher level model, but is limited to selecting behavior using simple planning approaches~\cite{jk96goms}. 
Simon's approach of \emph{bounded rationality}~\cite{s55behavioral} and related heuristic-based approaches~\cite{gs02bounded} model more complex planning, but with less precise predictions.

\section{Conclusions and Future Work}
We use planning to present the first formal semantics for determining when a sequence of actions is for a purpose.
In particular, our semantics uses an MDP-like model for planning, which allows us to automate auditing for both restrictive and prohibitive purpose requirements.  Furthermore, our semantics highlights that an action can be for a purpose even if that purpose is never achieved, a point present in philosophical works on the subject (e.g.,~\cite{t66action}), but whose ramifications on policy enforcement had been unexplored.  Lastly, our framework allows us to explain and compare previous methods of policy enforcement in terms of a formal semantics.

However, we recognize the limitations of this model: it imperfectly models human planning and only captures some forms of planning for multiple purposes.  Nevertheless, we believe the essence of our work is correct: an action is for a purpose if the actor selects to perform that action while planning for the purpose.  Future work will instantiate our semantic framework with more complete models of human planning.

Fundamentally, our work shows the difficulties of enforcement due to issues such as the tenable deniability of ulterior motives.  These difficulties justify policies prohibiting conflicts of interest and requiring the separation of duties despite possibly causing inefficiencies.  For example, many hospitals would err on the side of caution and disallow referral from a physician to his own private practice or require a second opinion to do so, thereby restraining the ulterior motive of profit.  Indeed, despite the maxim that \emph{privacy is security with a purpose}, due to these difficulties, purpose possibly plays the role of guidance in crafting more operational internal policies that organizations enforce rather than the role of a direct input to the formal auditing process itself.  In light of this possibility, one may view our work as a way to judge the quality of these operational policies relevant to the intent of the purpose requirements found in the actual privacy policy.

We further believe that our formalism may aid organizations in designing their processes to avoid the possibility of or to increase the detectability of policy violations.  For example, the organization can decrease violations by aligning employee incentives with the allowed purposes.

\paragraph{Acknowledgments}
We appreciate the discussions we have had with Lorrie Faith Cranor and Joseph Y. Halpern on this work.  We thank Dilsun Kaynar and Divya Sharma for many helpful comments on this paper.

\bibliographystyle{alphaurl}
\bibliography{purpose}


\onecolumn
\appendix


\section{Details of MDPs}

One may find a discussion of MDPs in most introductions to artificial intelligence (e.g.,~\cite{rn03artificial}).  For an MDP $m = \langle \mathcal{Q}, \mathcal{A}, t, r, \gamma\rangle$, the discount factor $\gamma$ accounts for the preference of people for receiving rewards sooner than later.  It may be thought of as similar to inflation.  We require that $\gamma < 1$ to ensure that the expected total discounted reward is bounded.  

The value of a state $q$ under a strategy $\sigma$ is
\[ V_m(\sigma, q) = \E \left[ \sum_{i=0}^\infty \gamma^i r(q_i,\sigma(q_i)) \right] \] 
The Bellman equation shows that 
\[ V_m(\sigma,q) = r(q,\sigma(q)) + \gamma \sum_{q' \in \mathcal{Q}} t(q,\sigma(q))(q') * V_m(\sigma,q') \]

A strategy $\sigma^*$ is optimal if and only if for all states $q$, $V_m(\sigma^*,q) = \max_{\sigma} V_m(\sigma,q)$.  At least one optimal policy always exists.  Furthermore, if $\sigma^*$ is optimal, then 
\[ \sigma^*(q) = \argmax_{a \in \mathcal{A}} \left[ r(s,a) + \gamma \sum_{q' \in \mathcal{Q}} t(q,\sigma(q))(q') * V_m(\sigma,q') \right] \]

\section{Proof of Theorem~\ref{thm:opt-not-empty}}

\newcommand{\subexec}{\triangleleft}
\newcommand{\subexeceq}{\trianglelefteq}

The proper sub-execution relation is a strict partial order.  This follows directly from the proper-subsequence relation $\subseq$ being a strict partial order.   We write $\subexec$ for \emph{proper sub-execution} and $\subexeceq$ for \emph{proper sub-execution or equal}.

Now, we show that $\substrg$ is also strict partial ordering.  
\begin{itemize}
\item Irreflexivity: for no $\sigma$ is $\sigma \substrg \sigma$.  For $\sigma \substrg \sigma$ to be true, there would have to exist a $\sigma \in \opt$ such that for at least one contingency $\kappa'$ and $q'$, $m(q',\kappa', \sigma')$ is a proper sub-execution of itself.  However, this is impossible since the sub-execution relation is strict partial order.

\item Asymmetry: for all $\sigma_1$ and $\sigma_2$, if $\sigma_1 \substrg \sigma_2$, then it is not the case that $\sigma_2 \substrg \sigma_1$.  
To show a contradiction, suppose $\sigma_1 \substrg \sigma_2$ and $\sigma_2 \substrg \sigma_1$ are both true.  It would have to be the case 
that for all contingencies $\kappa$ and states $q$, 
$m(q,\kappa, \sigma_1) \subexeceq m(q,\kappa, \sigma_2)$
and $m(q,\kappa, \sigma_2) \subexeceq m(q,\kappa, \sigma_1)$.
Since $\subexec$ is a strict partial order, this implies that for all $q$ and $\kappa$, $m(q,\kappa, \sigma_1) = m(q,\kappa, \sigma_2)$.  Thus, there cannot exist a contingency $\kappa'$ and state $q'$ such that
$m(q',\kappa', \sigma_2) \subexec m(q',\kappa', \sigma_1)$.
Then $\sigma_2 \substrg \sigma_1$ cannot be true, a contradiction.

\item Transitivity: for all $\sigma_1$, $\sigma_2$, and $\sigma_3$, if $\sigma_1 \substrg \sigma_2$ and $\sigma_2 \substrg \sigma_3$, then $\sigma_1 \substrg \sigma_3$.
Suppose $\sigma_1 \substrg \sigma_2$ and $\sigma_2 \substrg \sigma_3$.
Then for all for all contingencies $\kappa$ and states $q$, 
$m(q,\kappa, \sigma_1) \subexeceq m(q,\kappa, \sigma_2)$
and $m(q,\kappa, \sigma_2) \subexeceq m(q,\kappa, \sigma_3)$.  Since $\subexeceq$ has transitivity, this implies that $m(q,\kappa, \sigma_1) \subexeceq m(q,\kappa, \sigma_3)$ for all $\kappa$ and $q$.

Furthermore, it must be the case that there exists a contingency $\kappa'$ and state $q'$ such that $m(q',\kappa', \sigma_1) \subexec m(q',\kappa', \sigma_2)$.  From above, $m(q',\kappa', \sigma_2) \subexeceq m(q',\kappa', \sigma_3)$.  Thus, by the transitivity of $\subexeceq$, $m(q',\kappa', \sigma_1) \subexec m(q',\kappa', \sigma_3)$ as needed.
\end{itemize}

Since $\substrg$ is a strict partial ordering and $\mathcal{Q} \to \mathcal{A}$ is finite, $\mathcal{Q} \to \mathcal{A}$ is well-founded under $\substrg$.   $\mathcal{Q} \to \mathcal{A}$ being finite also means that $\opt(m)$ is finite.  It is also known to be non-empty~\cite{rn03artificial}.

Suppose $\opt^*(m)$ were empty.  This would mean for every $\sigma$ of $\opt(m)$, there exists $\sigma'$ in $\opt(m)$ such that $\sigma' \substrg \sigma$.
Since $\opt(m)$ is finite but non-empty, this could only happen if $\substrg$ contained cycles.  However, this is a contradiction since 
 $\substrg$ is a strict partial order and $\mathcal{Q} \to \mathcal{A}$ is well-founded under it.  Thus, $\opt^*(m)$ is not empty.

\section{Proofs about Useless States}

\begin{proposition}\label{prp:useless}
For all environment models $m$, sets $U$ such that $U \subseteq U_m$, strategies $\sigma$, and states $q$, $V_m(\sigma, q) \leq V_m(U(\sigma), q)$.
\end{proposition}
\begin{proof}

\newcommand*{\exec}{\mathsf{exec}}
Let $\exec(b)$ be all the executions with the behavior $b$ as a prefix.
Let $B_U$ be the set of all behaviors $b$ such that for some $j$, $b = [q_0, a_1, q_1 \ldots, q_{j}, a_{j+1}, q_{j+1}]$ such that $\langle q_{j}, a_{j+1}\rangle$ is in $U$ but for not $i < j$ is $\langle q_{i-1}, a_i\rangle$ in $U$.
We may use $B_U$ and $\exec(b)$ to partition the space of executions $E$.
Thus, 
\begin{align}
V_m(\sigma, q) 
&= \E \left[ \sum_{i=0}^\infty \gamma^i r(q_i,\sigma(q_i)) \right]\notag\\
&= \sum_{e \in E} \Pr[e | \sigma] \left[ \sum_{i=0}^\infty \gamma^i r(q_i,\sigma(q_i)) \right]\notag\\
&= \sum_{b \in B_U} \sum_{e \in \exec(b)} \Pr[e | \sigma] \left[ \sum_{i=0}^\infty \gamma^i r(q_i,\sigma(q_i)) \right] \label{ln:useless.inner-form}
\end{align}
(Note: as $E$ is uncountable, taking a summation over it is ill advised.  We could take an integral instead.  Alternatively, one could take the sum over executions of bounded length.  This will introduce an error term.  However, as the bound increases the magnitude of this term will drop exponentially fast due to the factor $\gamma$.  In essence, this is how most practical algorithms for solving MDPs operate.  See~\cite{rn03artificial}.)

For any $b$ in $B_U$, consider $e \in \exec(b)$.
Since $e$ is in $\exec(b)$, it must have the following form
\[ [q_0, a_1, q_1 \ldots, q_{j}, a_{j+1}, q_{j+1}, \ldots] \]
where $\langle q_{j}, a_{j+1}\rangle \in U$ but for $i < j$ is $\langle q_{i}, a_{i+1}\rangle \notin U$ where $b = [q_0, a_1, q_1 \ldots, q_{j}, a_{j+1}, q_{j+1}]$.

For $\sigma \in \strg(b)$, we reason as shown as follows.
\begin{align}
\sum_{e \in \exec(b)} \Pr[e | \sigma] & \left[ \sum_{i=0}^\infty \gamma^i r(q_i,\sigma(q_i)) \right] \notag\\
&= \sum_{e \in \exec(b)} \Pr[e | \sigma] \left[ \sum_{i=0}^{j-1} \gamma^i r(q_i,\sigma(q_i)) + \sum_{i=j}^\infty \gamma^i r(q_i,\sigma(q_i)) \right] \notag\\
&= \left[ \sum_{e \in \exec(b)}  \Pr[e | \sigma] \sum_{i=0}^{j-1}  \gamma^i r(q_i,\sigma(q_i)) \right] + \gamma^j  \left[ \sum_{e \in \exec(b)}  \Pr[e | \sigma]  \sum_{i=j}^\infty \gamma^{i-j} r(q_i,\sigma(q_i)) \right] \notag\\
&= \left[ \sum_{e \in \exec(b)}  \Pr[e | \sigma] \sum_{i=0}^{j-1} \gamma^i r(q_i,\sigma(q_i)) \right] + \gamma^j \E \left[ \sum_{i=j}^\infty \gamma^i r(q_i,\sigma(q_i)) \right] \notag\\
&= \left[ \sum_{e \in \exec(b)}  \Pr[e | \sigma] \sum_{i=0}^{j-1} \gamma^i r(q_i,\sigma(q_i)) \right] + \gamma^j V_m(\sigma,q_{j}) \label{ln:useless.make-form}
\end{align}
Furthermore,
\[
\sum_{e \in \exec(b)}  \Pr[e | \sigma] \sum_{i=0}^{j-1} \gamma^i r(q_i,\sigma(q_i))
= \Pr[b | \sigma] \sum_{i=0}^{j-1} \gamma^i r(q_i,\sigma(q_i))  \]
Thus, the left term is equal under $\sigma$ and $U(\sigma)$:
\begin{align}
\sum_{e \in \exec(b)}  \Pr[e | \sigma] \sum_{i=0}^{j-1} \gamma^i r(q_i,\sigma(q_i))
&=\Pr[b | \sigma] \sum_{i=0}^{j-1} \gamma^i r(q_i,\sigma(q_i)) \\
&=\Pr[b | \sigma] \sum_{i=0}^{j-1} \gamma^i r(q_i,U(\sigma)(q_i)) \label{ln:useless.same-in-front-i}\\
&= \sum_{e \in \exec(b)} \Pr[e | U(\sigma)] \sum_{i=0}^{j-1} \gamma^i r(q_i,U(\sigma)(q_i)) \label{ln:useless.same-in-front-ii}
\end{align}
where line~\ref{ln:useless.same-in-front-i} follows since $\sigma(q_i) = U(\sigma)(q_i)$ for $\langle q_{i}, a_{i+1}\rangle \notin U$.

Since $\langle q_{j}, a_{j+1}\rangle \in U$, we know that $Q_m(\sigma,q_{j},a_{j+1}) \leq 0$.  Furthermore, since $\sigma \in \strg(b)$, it is the case that $\sigma(q_{j}) = a_{j+1}$.  Thus, $V_m(\sigma, q_{j}) = Q_m(\sigma,q_{j},\sigma(q_{j})) \leq 0$.  Furthermore, since $\langle q_{j}, a_{j+1}\rangle \in U$, $V_m(U(\sigma), q_{j}) = Q_m(\sigma,q_{j},\nothing) = 0$.
Thus, we may conclude
\begin{align}
\sum_{e \in \exec(b)} \Pr[e | \sigma] & \left[ \sum_{i=0}^\infty \gamma^i r(q_i,\sigma(q_i)) \right] \notag\\
&= \left[ \sum_{e \in \exec(b)}  \Pr[e | \sigma] \sum_{i=0}^{j-1} \gamma^i r(q_i,\sigma(q_i)) \right] + \gamma^j V_m(\sigma,q_{j})\label{ln:useless.use-form-i}\\
&= \left[ \sum_{e \in \exec(b)}  \Pr[e | U(\sigma)] \sum_{i=0}^{j-1} \gamma^i r(q_i,U(\sigma)(q_i)) \right] + \gamma^j V_m(\sigma,q_{j}) \label{ln:useless.left}\\
&\leq \left[ \sum_{e \in \exec(b)}  \Pr[e | U(\sigma)] \sum_{i=0}^{j-1} \gamma^i r(q_i,U(\sigma(q_i))) \right] + \gamma^j V_m(U(\sigma),q_{j})\label{ln:useless.right}\\\\
&= \sum_{e \in \exec(b)} \Pr[e | U(\sigma)] \left[ \sum_{i=0}^\infty \gamma^i r(q_i,U(\sigma)(q_i)) \right]\label{ln:useless.use-form-ii}
\end{align}
where lines~\ref{ln:useless.use-form-i} and~\ref{ln:useless.use-form-ii} come from the reasoning 
leading to line~\ref{ln:useless.make-form}, and
line~\ref{ln:useless.left} comes from the reasoning leading to line~\ref{ln:useless.same-in-front-ii}.

Note that the above also trivially holds when $\sigma \notin \strg(b)$ since $\Pr[e | \sigma] = 0$ and  $\Pr[e | U(\sigma)] = 0$ for all $e \in \exec(b)$.  Thus, for all $\sigma$, we have
\begin{align}
\sum_{e \in \exec(b)} \Pr[e | \sigma] \left[ \sum_{i=0}^\infty \gamma^i r(q_i,\sigma(q_i)) \right]
&\leq \sum_{e \in \exec(b)} \Pr[e | U(\sigma)] \left[ \sum_{i=0}^\infty \gamma^i r(q_i,U(\sigma)(q_i)) \right] \label{eqn:useless.inner-terms-eq}
\end{align}

Thus,
\begin{align}
V_m(\sigma, q)
&= \sum_{b \in B_U} \sum_{e \in \exec(b)} \Pr[e | \sigma] \left[ \sum_{i=0}^\infty \gamma^i r(q_i,\sigma(q_i)) \right] \label{ln:useless.use-inner-form-i}\\
&\leq \sum_{b \in B_U} \sum_{e \in \exec(b)} \Pr[e | \sigma] \left[ \sum_{i=0}^\infty \gamma^i r(q_i,\sigma(q_i)) \right] \label{ln:useless.use-inner-terms-eq}\\
&= V_m(U(\sigma), q) \label{ln:useless.use-inner-form-ii}
\end{align}
where
line~\ref{ln:useless.use-inner-form-i} and~\ref{ln:useless.use-inner-form-ii} comes from the reasoning of line~\ref{ln:useless.inner-form}, and
line~\ref{ln:useless.use-inner-terms-eq} comes from equation~\ref{eqn:useless.inner-terms-eq}.
\end{proof}

\section{Proof of Lemma~\ref{lem:two-steps}}

First we prove that this $\log^{-1}(b) \cap \behv^*(m)$ in the lemma may be replaced with $\strg_m(b) \cap \opt^*(m)$.
Then, we prove the modified statement with two propositions.  We have one corresponding to the \emph{if} direction and one to the \emph{only if} direction.

\begin{proposition}
For environment models $m$, if for all observable behaviors $b$, $\log(b) = b$, then $\strg(b) \cap \opt^*(m)$ is empty if and only if $\log^{-1}(b) \cap \behv^*(m)$ is empty.
\end{proposition}
\begin{proof}
Since $\log^{-1}(b) = \{ b \}$, $\log^{-1}(b) \cap \behv^*(m)$ is empty if and only $b \notin \behv^*(m)$. 
$b$ is in $\behv^*(m)$ if and only if there exists a strategy $\sigma$ in $\opt^*(m)$ such that there exists a contingency $\kappa$, and a state $q$ such that $\be$ is a subsequence of $m(q, \kappa, \sigma)$.

For all $\sigma$ in $\opt^*(m)$, $\exists \kappa, q . \be \subseq m(q,\kappa,\sigma)$ is equivalent to $\forall i \in [0,n). \sigma(q_{i}) = a_{i+1}$ where $\be = [q_0, a_1, q_1, a_2, \ldots, a_n, q_n]$.
To see this, note $b$ was observed and, thus, it must have been produced by a contingency consistent with $m$.

$\forall i \in [0,n). \sigma(q_{i}) = a_{i+1}$ is equivalent to $\sigma \in \strg(b)$.
Thus, $b$ is in $\behv^*(m)$ if and only if there exists a strategy $\sigma$ in $\opt^*(m)$ such that $\sigma$ is in $\strg(b)$.
Thus, $\log^{-1}(b) \cap \behv^*(m)$ is not empty if and only if $\strg(b) \cap \opt^*(m)$ is not empty.
\end{proof}

\begin{proposition}
For all environment models $m$ and behaviors $\be = [q_0, a_1, q_1, \ldots, a_n, q_n]$,
$\strg(b) \cap \opt^*(m)$ is not empty if (1) for all $i$ such that $0 \leq i < n$, $\langle q_i, a_{i+1}\rangle \notin U_m$ and (2) $\strg(b) \cap \opt(m)$ is not empty.
\end{proposition}
\begin{proof}
Suppose the conditions (1) and (2) are true.
Since  $\strg(b) \cap \opt(m)$ is not empty, there exists some $\sigma_1$ in both of them.  Since $\sigma_1$ is in $\strg(b)$, for all $0 \leq i < n$, $\sigma_1(q_i) = a_{i+1}$.  Thus, by condition (2), $\langle q_i, \sigma_1(q_i)\rangle \notin U_m$.  This further implies that $a_{i+1}$ is not $\nothing$.

Let $\sigma_2 = U_m(\sigma_1)$.  $\sigma_2$ is in $\strg(b)$ because for all $0 \leq i < n$, $\sigma_1(q_i) = \sigma_2(q_i)$ since $\langle q_i, \sigma_1(q_i) \rangle \notin U_m$.
Furthermore, by Proposition~\ref{prp:useless}, for all $q$, $V_m(\sigma_1, q) \leq V_m(\sigma_2, q)$.  Thus, $\sigma_2$ is in $\opt(m)$ as well.

To show that $\sigma_2$ is also in $\opt^*(m)$, suppose it were not.  Since $\sigma_2$ is in $\opt(m)$, this implies that there exists $\sigma'$ in $\opt(m)$ such that $\sigma' \substrg \sigma_2$.  For this to be true, there must exist $\kappa'$ and state $q'$ such that $\beforenothing(m(q',\kappa',\sigma')) \subseq \beforenothing(m(q',\kappa',\sigma_2)$.  Thus, for some $i$, 
$m(q',\kappa',\sigma_2)$ must have the form $[q_0, a_1, q_1, \ldots, q_{i-1}, a_{i}, q_{i}, a_{i+1}, q_{i+1}, \ldots]$, and 
$m(q',\kappa',\sigma')$ must have the form  $[q_0, a_1, q_1, \ldots, q_{i-1}, a_{i}, q_{i}, \nothing, q_{i}, \ldots]$ 
where $a_{i+1}$ is not $\nothing$.
Since $\sigma_2(q_{i}) = a_{i+1}$, by the construction of $\sigma_2$, $\langle q_{i}, a_{i+1}\rangle$ is not in $U_m$. 
Thus, there exists some $\sigma_3$ such that $Q_m(\sigma_3, q_i, a_{i+1}) > 0$.

Since $\sigma_2$ is in $\opt(m)$, $Q_m(\sigma_2, q_i, a_{i+1}) \geq Q_m(\sigma_3, q_i, a_{i+1}) > 0$.
Thus, we have $V_m(\sigma_2,q_i) = Q_m(\sigma_2, q_i, a_{i+1}) > 0$.
However, $V_m(\sigma',q_i) = 0$ meaning that $\sigma'$ is not in $\opt(m)$, a contradiction.
\end{proof}

\begin{proposition}
For all environment models $m$ and behaviors $\be = [q_0, a_1, q_1, \ldots, a_n, q_n]$,
if $\strg(b) \cap \opt^*(m)$ is not empty, then
(1) for all $i$ such that $0 \leq i < n$, $\langle q_i, a_{i+1}\rangle \notin U_m$ and (2) $\strg(b) \cap \opt(m)$ is not empty.
\end{proposition}
\begin{proof}
Condition (2) follows from the fact that $\opt^*(m) \subseteq \opt(m)$.

To prove condition (1), suppose $\strg(b) \cap \opt^*(m)$ is not empty but condition (1) does not hold.  Then there exists $\sigma_1$ in  $\strg(b) \cap \opt^*(m)$.  Furthermore, there exists some $i'$ such that $\langle q_{i'}, a_{i'+1}\rangle \in U_m$.  Since $\sigma_1 \in \strg(b)$, it must be the case that for all $i < n$, $a_{i+1} = \sigma(q_{i})$.  Thus, $\sigma_1(q_{i'}) = a_{i'+1}$.
By Proposition~\ref{prp:useless}, for all $q$, $V_m(\sigma_1, q) \leq V_m(U_m(\sigma_1), q)$. 
Furthermore, $U(\sigma_1) \substrg \sigma_1$.  To see this, recall that $U_m$ is not empty. Thus, any contingency $\kappa'$ that results in state $q_{i'}$, $m(q_0,\kappa',U_m(\sigma_1)) \subseq m(q_0,\kappa',\sigma_1)$ since only $U_m(\sigma_1)$ does nothing at $q_{i'}$.
For $\kappa$ that do not lead to $q_{i'}$, the two executions will be the same.

Since $U_m(\sigma_1) \substrg \sigma_1$ and $U(\sigma_1)$ is in $\opt(m)$, $\sigma_1$ cannot be in $\opt^*(b)$, a contradiction.
\end{proof}

\section{Proof of Lemma~\ref{lem:useless-check}}

If $\langle q, a\rangle$ is in $U_m$, then $a \neq \nothing$ and for all strategies $\sigma$, $Q_m(\sigma,q,a) \leq 0$.  Thus, the lemma is true if the following is true: $Q^*(q,a) \leq 0$ iff $\forall \sigma . Q_m(\sigma,q,a) \leq 0$.

To show this, note that
$\forall \sigma . Q_m(\sigma,q,a) \leq 0$ iff $\max_{\sigma} Q_m(\sigma,q,a) \leq 0$.
Furthermore,
\begin{align*}
\max_{\sigma} Q_m(\sigma,q,a)
&= \max_{\sigma} r(q,a) + \gamma \sum_{q'} t(q,a)(q') * V_m(\sigma, q')\\
&= r(q,a) + \gamma \sum_{q'} t(q,a)(q') *  \max_{\sigma} V_m(\sigma, q')\\
&= r(q,a) + \gamma \sum_{q'} t(q,a)(q') * V^*(q')\\
&= Q^*_m(q,a)
\end{align*}
Thus, $\forall \sigma . Q_m(\sigma,q,a) \leq 0$ iff  $Q^*_m(q,a) \leq 0$.

\section{Properties of $\mathsf{fix}$}

\begin{proposition}\label{prp:fix-less}
For all environment models $m$, strategies $\sigma$, and states $q$, $V_{\fix(m,b)}(\sigma,q) \leq V_m(\sigma,q)$.
\end{proposition}
\begin{proof}
Let $m = \langle \mathcal{Q}, \mathcal{A}, t, r, \gamma\rangle$ and $\fix(m,b) = \langle \mathcal{Q}, \mathcal{A}, t, r', \gamma\rangle$.
\begin{align}
V_m(\sigma, q) 
&= \E \left[ \sum_{i=0}^\infty \gamma^i r(q_i,\sigma(q_i)) \right]\\
&\leq \E \left[ \sum_{i=0}^\infty \gamma^i r'(q_i,\sigma(q_i)) \right] \label{ln:fix-less.key}\\
&= V_{\fix(m,b)}(\sigma, q) 
\end{align}
where line~\ref{ln:fix-less.key} follows from the fact that for all $q$ and $a$, $r'(q,a) \leq r(q,a)$.
\end{proof}

\begin{proposition}\label{prp:fix-same}
For all environment models $m$, behaviors $b$, $\sigma \in \strg(b)$, and states $q$, $V_{\fix(m,b)}(\sigma,q) = V_m(\sigma,q)$
\end{proposition}
\begin{proof}
Let $m = \langle \mathcal{Q}, \mathcal{A}, t, r, \gamma\rangle$ and $\fix(m,b) = \langle \mathcal{Q}, \mathcal{A}, t, r', \gamma\rangle$.
Let $b = [q_0, a_1, q_1, \ldots, a_n, q_n]$.  

Since $\sigma$ is in $\strg(b)$, for all $i$ such that $0 \leq i < n$, $\sigma(q_i) = a_{i+1}$.  Thus, $r'(q_i,a_{i+1}) = r(q_i,a_{i+1})$.
For all $q$ that is not equal to $q_i$ for any $i$, $r'(q,a) = r(q,a)$ for all $a$.
Thus, for all $a$ and $q$, $r'(q,\sigma(q)) = r(q,\sigma(q))$.
This implies
\begin{align*}
V_m(\sigma, q) 
&= \E \left[ \sum_{i=0}^\infty \gamma^i r(q_i,\sigma(q_i)) \right]\\
&= \E \left[ \sum_{i=0}^\infty \gamma^i r'(q_i,\sigma(q_i)) \right]\\
&= V_{\fix(m,b)}(\sigma, q) 
\end{align*}
\end{proof}

\begin{proposition}\label{prp:fix-as-good}
For all environment models $m$, behaviors $b$, and $\sigma_1 \notin \strg(b)$, there exists a $\sigma_2 \in \strg(b)$ such that for all states $q$, $V_{\fix(m,b)}(\sigma_1,q) \leq V_{\fix(m,b)}(\sigma_2,q)$.
\end{proposition}
\begin{proof}
Let $\fix(m,b) = \langle \mathcal{Q}, \mathcal{A}, t, r', \gamma\rangle$.
Let $b = [q_0, a_1, q_1, \ldots, a_n, q_n]$.  
Since $\sigma_1$ is not in $\strg(b)$, there must exist some $i$ such that $\sigma_1(q_i) \neq a_{i+1}$.  
Let the set $I$ hold all such indexes $i$: $I = \set{i \in [0,n)}{\sigma_1(q_i) \neq a_{i+1}}$.  
Let $\sigma_2$ be the strategy such that $\sigma_2(q) = a_{i+1}$ if $q = q_i$ for some $i \in I$ and $\sigma_2(q) = \sigma_1(q)$ otherwise.
By construction, $\sigma_2$ is in $\strg(b)$.

By the construction of $\fix(m,b)$, for all $i \in I$, $r'(q_i,\sigma_1(q_i)) = -\omega \leq r'(q_i, a_{i+1}) = r'(q_i, \sigma_2(q_i))$.  Thus, for all $q$, $r'(q,\sigma_1(q)) \leq r'(q,\sigma_2(q))$.  Thus, for all states $q$, $V_{\fix(m,b)}(\sigma_1,q) \leq V_{\fix(m,b)}(\sigma_2,q)$.
\end{proof}

\begin{proposition}\label{prp:fix-sometimes-bad}
For all environment models $m$, behaviors $b$, $\sigma_1 \notin \strg(b)$, and $\sigma_2 \in \strg(b)$, there exists a state $q$ such that $V_{\fix(m,b)}(\sigma_1,q) < V_{\fix(m,b)}(\sigma_2,q)$.
\end{proposition}
\begin{proof}
Let $b = [q_0, a_1, q_1, \ldots, a_n, q_n]$.  
Since $\sigma_1$ is not in $\strg(b)$, there must exist some $i$ such that $\sigma_1(q_i) \neq a_{i+1}$.  By the construction of $\fix(m,b)$, $r'(q_i,\sigma_1(q_i)) = -\omega$.  
Recall that $\omega > 2r^*/(1-\gamma)$ where $r^*$ is the reward with the largest magnitude.  Thus,
\begin{align}
V_{\fix(m,b)}(\sigma_1,q_i)
&= r(q_i,\sigma_1(q_i)) + \gamma \sum_{q'} t(q_i,\sigma_1(q_i))(q')*V_m(\sigma,q')\\
&= -\omega + \gamma \sum_{q'} t(q_i,\sigma_1(q_i))(q')*V_m(\sigma,q')\\
&\leq -\omega + \gamma \sum_{q'} t(q_i,\sigma_1(q_i))(q') * r^*/(1-\gamma)\\
&= -\omega + \gamma * r^*/(1-\gamma) \label{ln:fix-sometimes-bad.prob}\\
&\leq -\omega + r^*/(1-\gamma)\\
&< -[2r^*/(1-\gamma)] + r^*/(1-\gamma)\\
&= -r^*/(1-\gamma)\\
&\leq V_{m}(\sigma_2,q)\label{ln:fix-sometimes-bad.normal}\\
&= V_{\fix(m,b)}(\sigma_2,q)\label{ln:fix-sometimes-bad.same}
\end{align}
where line~\ref{ln:fix-sometimes-bad.prob} follows from $t(q,\sigma_1(q)$ being a probability distribution over states,
line~\ref{ln:fix-sometimes-bad.normal} follows from the definition of $r^*$ and known bounds (e.g.,~\cite{rn03artificial}),
and line~\ref{ln:fix-sometimes-bad.same} follows from by Proposition~\ref{prp:fix-same}.
\end{proof}

\begin{proposition}\label{prp:opt-in-strg}
For all environment models $m$ and behaviors $b$,
$\opt(\fix(m,b))$ is a subset of $\strg(b)$.
\end{proposition}
\begin{proof}
Suppose $\sigma_1$ were not in $\strg(b)$.  
By Proposition~\ref{prp:fix-sometimes-bad}, for all $\sigma_2 \in \strg(b)$, there exists a state $q$ such that $V_{\fix(m,b)}(\sigma_1,q) < V_{\fix(m,b)}(\sigma_2,q)$.
Thus, $\sigma_1$ is not in $\opt(\fix(m,b))$.
\end{proof}

\begin{proposition}\label{prp:opt-equal}
For all environment models $m$, behaviors $b$, and strategies $\sigma$ in $\opt(\fix(m,b))$, $V_{\fix(m,b)}(\sigma, q) = V_m(\sigma, q)$.
\end{proposition}
\begin{proof}
Let $\sigma$ be in $\opt(\fix(m,b))$.
$\sigma$ must be in $\strg(m,b)$ by Proposition~\ref{prp:opt-in-strg}.
Thus, $V_{\fix(m,b)}(\sigma, q) = V_m(\sigma, q)$ by Proposition~\ref{prp:fix-same}.
\end{proof}

\section{Proof of Lemma~\ref{lem:fix-value-check}}

\newcommand*{\strglessopt}[2]{\textsf{strg-opt}({#1},{#2})}

This lemma follows directly from the Propositions~\ref{prp:fix-works} and~\ref{prp:value-check} below.

\begin{proposition}\label{prp:fix-works}
For all environment models $m$ and behaviors $b$,
$\strg(b) \cap \opt(m) = \opt(\fix(m,b)) \cap \opt(m)$.
\end{proposition}
\begin{proof}
Consider the set $\strglessopt{m}{b} = \strg(b) \setless \opt(\fix(m,b))$.
For all $\sigma$ in $\strglessopt{m}{b}$, $\sigma$ is in $\strg(b)$ but not $\opt(\fix(m,b))$.
By being in $\strg(b)$, $V_{\fix(m,b)}(\sigma, q) = V_m(\sigma, q)$ by Proposition~\ref{prp:fix-same}.
Thus, since $\sigma$ is not in $\opt(\fix(m,b))$, $\sigma$ is not in $\opt(m)$ either by Proposition~\ref{prp:fix-less}.
This means that $\strglessopt{m}{b} \cap \opt(m)$ is empty.

Furthermore, $\opt(\fix(m,b)) \subseteq \strg(b)$ by Proposition~\ref{prp:opt-in-strg}.  Thus,
$\strg(b)  = \opt(\fix(m,b)) \cup (\strg(b) \setless \opt(\fix(m,b))) = \opt(\fix(m,b)) \cup \strglessopt{m}{b}$.
Thus,
\begin{align*}
\strg(b) \cap \opt(m)
&= (\opt(\fix(m,b)) \cup \strglessopt{m}{b}) \cap \opt(m) \\ 
&= (\opt(\fix(m,b)) \cap \opt(m)) \cup (\strglessopt{m}{b} \cap \opt(m))\\
&= \opt(\fix(m,b)) \cap \opt(m) \\
\end{align*}
\end{proof}

\begin{proposition}\label{prp:value-check}
For all environment models $m$ and behaviors $b$,
$\opt(m) \cap \opt(\fix(m,\be))$ is empty if and only if for all $q$, $\max_{\sigma} V_{\fix(m,\be)}(\sigma,q) \neq \max_{\sigma} V_m(\sigma, q)$.
\end{proposition}
\begin{proof}
Suppose that $\opt(m) \cap \opt(\fix(m,\be))$ is not empty.  Then there exists $\sigma^*$ in both of them.  
Thus, 
\begin{align}
\max_\sigma V_{\fix(m,\be)}(\sigma, q)
&= V_{\fix(m,\be)}(\sigma^*,q) \label{ln:value-check.sigma-opt-i}\\
&= V_m(\sigma^*, q) \label{ln:value-check.use-opt-equal-i}\\
&= \max_{\sigma} V_m(\sigma, q)\label{ln:value-check.sigma-opt-ii}
\end{align}
where line~\ref{ln:value-check.use-opt-equal-i} follows from Proposition~\ref{prp:opt-equal} and lines~\ref{ln:value-check.sigma-opt-i} and~\ref{ln:value-check.sigma-opt-ii} follow from $\sigma^*$ being in both $\opt(m)$ and $\opt(\fix(m,\be))$.

Suppose that for all $q$, $\max_{\sigma} V_m(\sigma, q_0) = \max_{\sigma} V_{\fix(m,\be)}(\sigma, q_0)$.
Let $\sigma^*$ be in $\opt(\fix(m,\be))$.
For all $q$,
\begin{align}
V_m(\sigma^*, q)
&= V_{\fix(m,\be)}(\sigma^*,q) \label{ln:value-check.use-opt-equal-ii}\\
&= \max_{\sigma} V_{\fix(m,\be)}(\sigma, q_0) \label{ln:value-check.its-opt}\\
&= \max_{\sigma} V_m(\sigma, q)
\end{align}
where line~\ref{ln:value-check.use-opt-equal-ii} follows from Proposition~\ref{prp:opt-equal} and line~\ref{ln:value-check.its-opt} from $\sigma^* \in \opt(\fix(m,\be))$.
Thus, $\sigma^*$ is in $\opt(m)$ and $\opt(m) \cap \opt(\fix(m,\be))$ is not empty.
\end{proof}

\section{Proof of Theorem~\ref{thm:algo-correct}}

Line 05 will return \emph{true} if there exists $i$ such that $a_{i+1} \neq \nothing$ and $Q^*_m(q_i,a_{i+1}) \leq 0$.  By Lemma~\ref{lem:useless-check}, this implies that $\langle q_i, a_{i+1}$ is in $U_m$.  By Lemma~\ref{lem:two-steps}, this implies that $\mathsf{log}^{-1}(b) \cap \behv^*(m)$ is empty under condition (1).

Lines 06--16 constructs $m' = \fix(m,b)$.  It constructs $r'$ from $r$ by first setting $r' = r$.  On lines 13--16, it then sets $r'(q_i,k)$ to be $-\omega$ for all $k$ such that $k \neq a_{i+1}$.  Thus, $r'(q_i,a_{i+1})$ will be left as $r(q_i,a_{i+1})$ as needed.

If Line 05 does not Line 21 will return \emph{false} if there exists $j$ such that $V^*_m(j) = V^*_m(j)$.
In this case, it cannot be that for all $q$, $\max_{\sigma} V_m(\sigma, q_0) = \max_{\sigma} V_{\fix(m,\be)}(\sigma, q_0)$.  Thus, by Lemma~\ref{lem:fix-value-check}, $\strg(b) \cap \opt(m)$ is not empty and condition (2) is false of Lemma~\ref{lem:two-steps}.  Since the function would had returned already at Line 05 if condition (1) were true, we know it is false.   Thus, by Lemma~\ref{lem:two-steps},  $\mathsf{log}^{-1}(b) \cap \behv^*(m)$ is not empty.

If Line 22 is reached, \emph{true} is returned.  This can only happen if condition (2) is true.  This implies that $\mathsf{log}^{-1}(b) \cap \behv^*(m)$ is empty by Lemma~\ref{lem:two-steps}.

Thus, the algorithm is correct whether it returns \emph{true} or \emph{false}.

\end{document}